\documentclass{article}

\usepackage[english]{babel}
\usepackage[letterpaper,top=1cm,bottom=2.5cm,left=3cm,right=3cm,marginparwidth=3.5cm]{geometry}

\usepackage{float}
\usepackage{pgfplots}
\pgfplotsset{width=10cm,compat=1.15}
\usepackage{pdflscape}
\usepackage{colortbl}
\usepackage{booktabs}
\usepackage{algorithm}
\usepackage{algpseudocode}
\usepackage{amsmath}
\usepackage{amssymb}
\usepackage{amsthm}
\usepackage{amsfonts}
\newtheorem{theorem}{Theorem}
\newtheorem{remark}{Remark}
\newtheorem{proposition}{Proposition}
\newtheorem{assumption}{Assumption}
\usepackage{threeparttable}  
\usepackage{graphicx}
\usepackage{subcaption} 
\usepackage{xcolor}
\usepackage{bm}
\usepackage{lipsum} 
\usepackage{titling} 
\usepackage[authoryear,longnamesfirst]{natbib}
\usepackage{appendix} 

\usepackage{indentfirst}

\usepackage[colorlinks=true, allcolors=blue]{hyperref}

\title{Recovering Counterfactual Distributions via Wasserstein GANs}

\author{
    Xinran Liu \\ 
    Department of Economics \\
    University of California, Riverside \\
    \texttt{xinran.liu@email.ucr.edu} 
}

\date{January 2026}
\newcommand{\REQUIRE}{\Require}

\newcommand{\COMMENT}[1]{\Comment{#1}}
\begin{document}
\maketitle
\begin{abstract}
Standard Distributional Synthetic Controls (DSC) estimate counterfactual distributions by minimizing the Euclidean $L_2$ distance between quantile functions. We demonstrate that this geometric reliance renders estimators fragile: they lack informative gradients under support mismatch and produce structural artifacts when outcomes are multimodal. This paper proposes a robust estimator grounded in Optimal Transport (OT). We construct the synthetic control by minimizing the Wasserstein-1 distance between probability measures, implemented via a Wasserstein Generative Adversarial Network (WGAN). We establish the formal point identification of synthetic weights under an affine independence condition on the donor pool. Monte Carlo simulations confirm that while standard estimators exhibit catastrophic variance explosions under heavy-tailed contamination and support mismatch, our WGAN-based approach remains consistent and stable. Furthermore, we show that our measure-based method correctly recovers complex bimodal mixtures where traditional quantile averaging fails structurally.
\end{abstract}

\section{Introduction}
\label{intro}

Since its inception by \cite{abadie2003economic} and subsequent formalization by \cite{abadie2010synthetic}, the Synthetic Control Method has emerged as the canonical framework for comparative case studies in empirical economics. By constructing a counterfactual unit as a convex combination of untreated donors, the method provides a transparent, data-driven mechanism to estimate treatment effects. While the standard framework focuses exclusively on mean outcomes, recent policy evaluations increasingly demand an understanding of distributional impacts. Whether assessing the effect of minimum wage laws on income inequality or the impact of trade liberalization on firm productivity dispersion, structural interest often lies in the tails or the shape of the distribution rather than its first moment. Addressing this imperative, recent scholarship has advanced the concept of DSC to recover the full counterfactual probability law of the treated unit \citep{gunsilius2023distributional}.

Despite these conceptual advances, existing implementations remain vulnerable to the limitations of the Euclidean metrics they employ. Standard approaches typically estimate counterfactuals by minimizing the integrated squared error between the quantile functions of the treated and synthetic units. This strategy relies on pointwise vertical comparisons and implicitly assumes that the counterfactual distribution is a location-scale transformation of the donors. Such a restriction is non-trivial; it imposes a form of rank invariance that prevents the estimator from recovering complex morphological changes, such as a unimodal distribution evolving into a bimodal one. Furthermore, in settings characterized by support mismatch—where the outcomes of the treated unit lie outside the support of the donor distributions—the squared error objective becomes locally constant. This lack of informative gradients renders gradient-based optimization ineffective and leads to severe estimator instability.

Our methodological departure from Euclidean metrics to OT is primarily motivated by a pervasive challenge in financial econometrics: characterizing the geometry of risk under structural breaks. Financial time series—unlike smooth macroeconomic aggregates—are frequently characterized by heavy-tailed contamination, disjoint supports due to liquidity gaps, and complex multimodal distributions during regime switches. In such regimes, standard distributional estimators that rely on quantile averaging or probability density matching often yield vanishing gradients or spurious unimodal artifacts, failing to capture the true latent factor structure. By reformulating the Synthetic Control problem as a Wasserstein minimization via a WGAN, we provide a robust estimator capable of recovering counterfactual probability laws even when the post-shock distribution shares no common support with the pre-shock history. While we demonstrate this capability through a canonical policy experiment—the Kansas tax cuts, which serves as an ideal stress test involving a massive, exogenous distributional shift—the proposed framework is broadly applicable to asset pricing and risk management, particularly for recovering counterfactual return distributions and tail dependencies in the presence of idiosyncratic shocks.

To resolve these impediments, we propose a robust estimator grounded in the theory of OT. We identify synthetic control weights by minimizing the Wasserstein-1 distance between the treated and synthetic distributions. The departure from Euclidean metrics to the Wasserstein distance represents a fundamental shift in how discrepancies are quantified. While Euclidean metrics penalize deviations at fixed coordinate points, the Wasserstein distance measures the minimum cost required to transport probability mass from one distribution to another. This metric respects the underlying geometry of the outcome space and induces a topology that preserves gradient information even when distributions have disjoint supports.

The economic intuition behind this shift is straightforward. Consider a scenario where a policy shifts the entire outcome distribution of a treated unit by a constant amount relative to a control unit. A Euclidean metric comparing densities or quantile functions at fixed points would register this merely as a maximum error, failing to capture the proximity of the two distributions. In contrast, the Wasserstein distance linearly reflects the magnitude of this shift, correctly identifying the structural similarity between the units. This property allows our estimator to effectively search for donors that are structurally closest to the treated unit based on the geometry of the data, rather than failing due to a lack of local density overlap. Consequently, our approach relaxes the rigidity of quantile averaging and permits the construction of synthetic counterfactuals that exhibit complex shapes, such as mixtures, which are not present in any single donor unit.

Our work contributes to three distinct strands of literature. First, we extend the methodological frontier of the SCM. Following the seminal work of \cite{abadie2010synthetic}, numerous extensions have been proposed to improve weight selection, including penalized estimators and matrix completion methods \citep{athey2021matrix, ben2021augmented}. While these innovations enhance robustness for mean outcomes, they do not address the estimation of full distributions. Second, we connect with the growing literature on OT in economics. \cite{galichon2016optimal} demonstrated the power of OT for structural estimation in matching models, yet its application to causal inference and time-series counterfactuals remains nascent. We bridge this gap by showing that the geometry induced by OT is uniquely suited for the Synthetic Control framework. Third, we add to the integration of machine learning into econometrics. Unlike approaches that utilize machine learning purely for prediction \citep{chernozhukov2018double}, we employ deep learning architectures to solve a structural minimization problem defined by a theoretically grounded statistical metric.

The primary theoretical contributions of this paper are the asymptotic theory and a formal identification result for synthetic controls in the space of probability measures. We demonstrate that the validity of our estimator does not require the restrictive rank-invariance assumptions inherent in quantile-based methods. Instead, we posit a latent factor model where the evolution of the outcome distribution is governed by unobserved time-varying common factors and unit-specific factor loadings. We show that under a scaled isometry condition, a synthetic control that successfully matches the pre-treatment distribution of the treated unit in terms of the Wasserstein distance effectively recovers the underlying factor loadings. This result generalizes the classic parallel trends assumption to the Wasserstein space, providing a rigorous justification for using linear mixtures of donors to approximate the counterfactual probability law of the treated unit.

Computationally, estimating the Wasserstein distance in high-dimensional or continuous spaces presents significant challenges, often requiring linear programming solutions that are cubically expensive in the sample size. To render this estimation feasible, we introduce a computational framework based on the Kantorovich-Rubinstein duality. Rather than solving the primal transport problem directly, we estimate the dual potentials using a WGAN. In this architecture, a critic network approximates the Lipschitz-continuous function that maximizes the transport cost, while a generator network—representing the synthetic weights—minimizes it. This adversarial formulation bypasses the sorting bottlenecks of traditional quantile methods and allows for the scalable, differentiable estimation of weights even when the outcome variables are multivariate. By solving the dual problem, our estimator robustly handles continuous supports and provides stable gradients, ensuring convergence even in regimes where standard density-matching techniques fail.

We corroborate our theoretical and computational arguments through a series of Monte Carlo simulations designed to stress-test the estimator under pathological conditions often encountered in economic data. Our results demonstrate that the WGAN-based solver significantly outperforms standard $L_2$ benchmarks in terms of both stability and accuracy. Specifically, in scenarios involving heavy-tailed data contamination and severe support mismatch, the standard quantile-based estimator exhibits substantial variance and bias, whereas our Wasserstein-based method remains robust. We further distinguish our approach by its ability to handle structural complexity. In a simulation where the true counterfactual involves a transition from a unimodal to a bimodal distribution—mimicking phenomena such as labor market polarization—we show that while quantile-averaging methods structurally fail and yield a unimodal artifact, our mixture-based approach successfully recovers the distinct modes. Additionally, we provide a proof-of-concept demonstration in a multivariate setting, confirming that our adversarial framework scales effectively to higher-dimensional outcomes where traditional sorting-based quantile methods are computationally intractable.

The remainder of the paper is organized as follows. Section 2 reviews the relevant literature. Section 3 outlines the econometric framework and defines the notation. Section 4 details the estimation strategy, introducing the WGAN architecture and the optimization algorithm based on the dual formulation of OT. Section 5 establishes the asymptotic theory and the formal identification result. Section 6 reports the results of our Monte Carlo simulations to evaluate finite-sample performance. Section 7 presents an empirical application. Section 8 concludes. All mathematical proofs and implementation details are provided in the Appendix.

\section{Literature Review}

Since its formalization by \cite{abadie2010synthetic}, the Synthetic Control Method has become a cornerstone of comparative case studies in empirical economics. The framework's ability to construct data-driven counterfactuals has catalyzed a rich methodological literature aimed at enhancing the robustness of average treatment effect estimates. Significant extensions include the introduction of penalized estimators to mitigate overfitting in high-dimensional donor pools \citep{abadie2021penalized} and the augmented synthetic control method, which corrects for regularization bias \citep{ben2021augmented}. Further integrating the method with modern causal inference tools, \cite{arkhangelsky2021synthetic} developed synthetic difference-in-differences to leverage both unit and time fixed effects, while \cite{athey2021matrix} reformulated the counterfactual imputation problem via matrix completion. Collectively, these contributions have substantially refined the estimation of conditional means $\mathbb{E}[Y]$, establishing a rigorous standard for aggregate policy evaluation.

Complementing these aggregate approaches, a growing literature emphasizes the necessity of evaluating distributional impacts. As articulated in the program evaluation literature by \cite{bitler2006mean} and \cite{firpo2009unconditional}, structural parameters of interest—such as inequality measures or tail risks—often depend on the full probability law rather than the first moment. Addressing this imperative, \cite{gunsilius2023distributional} pioneered the extension of SCM to the space of probability measures. This framework constructs the counterfactual distribution by minimizing the $L_2$ distance between the quantile function of the treated unit and a weighted average of donor quantile functions. Subsequent work has further solidified this quantile-based paradigm: \cite{zhang2024asymptotic} established the asymptotic optimality of such estimators under general regularity conditions, while \cite{kato2023asymptotically} introduced distribution-matching variations designed to enhance finite-sample performance. These contributions have successfully generalized the logic of synthetic controls to distributional objects under the assumption that the counterfactual preserves the rank structure of the donors.

Despite these theoretical strides, existing implementations remain vulnerable to data imperfections common in empirical settings. A primary concern arises from the reliance on squared-error loss functions. As noted in the robust statistics literature \citep{huber1992robust}, such metrics are highly sensitive to contamination, implying that tail outliers can disproportionately distort weight estimation. A second, more fundamental challenge pertains to the geometry of the support. The validity of standard quantile averaging is predicated on the assumption that the target distribution lies within the convex hull of the donors. In scenarios characterized by support mismatch—where the target distribution is disjoint from the donors—metrics based on pointwise vertical discrepancies fail to capture the underlying proximity, often resulting in uninformative gradients. These geometric limitations suggest the need for a metric that respects the topology of the outcome space, motivating our departure from Euclidean distances.

Our methodology aligns with the geometric turn in recent econometric theory, characterized by the increasing integration of OT principles into structural modeling. While classically rooted in pure mathematics, OT was systematically introduced to economics by \cite{galichon2016optimal}, primarily as a rigorous tool for analyzing matching markets and hedonic models. The central object of this theory—the Wasserstein distance—offers a metric for probability measures that explicitly respects the underlying geometry of the outcome space. This geometric sensitivity provides a decisive advantage in settings where standard information-theoretic divergences, such as the Kullback-Leibler divergence or the Total Variation distance, fail to provide meaningful comparisons for distributions with disjoint supports. Consequently, the Wasserstein metric has proven instrumental in facilitating advances across a spectrum of econometric challenges, ranging from distributionally robust optimization \citep{blanchet2019quantifying} to partial identification in incomplete models.

More pertinent to our framework is the subsequent shift from employing OT solely for identification to utilizing it for statistical inference. A growing literature on Wasserstein Minimum Distance Estimation (W-MDE) demonstrates that matching model-implied distributions to data via transport costs yields consistent estimates, offering a robust alternative in settings where the likelihood function is intractable or singular \citep{bernton2019parameter}. This inference paradigm shares a fundamental conceptual link with Generative Adversarial Networks (GANs). In the econometric domain, this connection has been formalized by \cite{kaji2023adversarial}, who establish that adversarial estimators can be interpreted as minimizing an Integral Probability Metric (IPM) defined by a discriminator class. From a statistical perspective, \cite{liang2021how} further provides minimax rates of convergence for distribution estimation under this framework, offering theoretical guarantees for non-parametric estimation. Collectively, these results provide the necessary foundation for utilizing neural networks in causal inference: rather than functioning as opaque ``black box'' generators, the adversarial architecture serves as a flexible, data-driven device for moment selection, ensuring that the synthetic controls satisfy the geometric constraints of the outcome space.

Despite these significant strides in cross-sectional distribution estimation, the application of OT-based adversarial inference to panel data and dynamic causal frameworks remains underexplored. While \cite{gunsilius2023distributional} explicitly acknowledges the theoretical appeal of Wasserstein metrics for synthetic controls—noting their ability to capture geometric proximity beyond rank invariance—no feasible computational strategy has yet been proposed to operationalize this concept for trajectory reweighting. Existing extensions of SCM that incorporate machine learning, such as \cite{athey2021matrix}, predominantly focus on low-rank approximations for mean outcomes rather than recovering full distributional counterfactuals via geometric transport. This paper bridges this specific gap. By adapting the Wasserstein GAN architecture to the unique structure of the synthetic control problem, we leverage the geometric robustness of OT to resolve the support mismatch failure that limits existing quantile-based estimators. This contribution effectively extends the practical reach of adversarial inference from static density estimation to the dynamic, counterfactual setting of comparative case studies.

\subsection{Deep Generative Models for Causal Inference}

Finally, our work contributes to the rapidly expanding interface between deep learning and causal inference. The integration of machine learning into econometrics has historically proceeded through two distinct waves. The first wave leveraged regularization techniques to manage high-dimensional controls and improve predictive accuracy \citep{belloni2014high, mullainathan2017does}. This foundation paved the way for a second wave focused on valid statistical inference, best exemplified by the Double Machine Learning framework of \cite{chernozhukov2018double} and the Generalized Random Forests of \cite{wager2018estimation}. These methods utilize flexible function approximators to handle nuisance parameters, thereby enabling the unbiased estimation of average and heterogeneous treatment effects. However, these canonical approaches operate primarily within the paradigm of discriminative learning, targeting conditional expectations $\mathbb{E}[Y|X]$ while leaving the modeling of complex, high-dimensional counterfactual distributions largely unaddressed.

To capture distributional heterogeneity beyond the conditional mean, recent scholarship has increasingly turned to generative modeling. Unlike discriminative approaches, generative frameworks aim to learn the underlying probability density of the potential outcomes. Notable contributions in this vein include the GANITE framework of \cite{yoon2018ganite}, which employs GANs to impute missing counterfactuals in cross-sectional data, and \cite{hartford2017deep}, who introduced Deep IV to address endogeneity via deep neural networks. While these studies demonstrate that adversarial training can effectively recover complex data-generating processes, the application of deep generative models to the Synthetic Control framework remains nascent. Current machine learning extensions of SCM predominantly focus on low-rank approximations; for instance, \cite{athey2021matrix} formulate the counterfactual estimation as a matrix completion problem using nuclear norm regularization. While powerful, such methods rely heavily on the assumption that the matrix of potential outcomes possesses a low-rank structure. They do not explicitly model the geometric transport of probability mass between distributions. Our work bridges this divide by proposing a WGAN-based solver that moves beyond the low-rank assumption, offering a methodology capable of preserving the distributional properties of the trajectory even when the data exhibits complex, non-linear structures.

\section{The Econometric Framework}\label{economicframe}

This section formalizes the problem of recovering counterfactual distributions. We depart from Euclidean metrics on quantile functions \citep{gunsilius2023distributional} by grounding our framework in the geometry of OT. This shift enables causal identification even when the support of the target distribution is disjoint from the donor pool, while accommodating granular, micro-level data structures.

\subsection{Setup and Notation}

Consider a panel of $J+1$ units observed over $T$ periods. Let unit $i=1$ denote the treated unit exposed to an intervention after period $T_0$, where $1 \le T_0 < T$. The set $\mathcal{D} = \{2, \dots, J+1\}$ serves as the donor pool of units that remain untreated throughout the observation window.

Unlike standard synthetic control methods that focus on point estimates $\mathbb{E}[Y_{it}]$, our object of interest is the full probability law of the outcome. We utilize the full granular information available at the micro-level. For each unit $i$ and period $t$, we observe a sample of $M_{it}$ independent observations, denoted by $\{y_{it}^{(k)}\}_{k=1}^{M_{it}}$. We define the observed outcome $P_{it} \in \mathcal{P}_1(\mathcal{X})$ as the empirical measure:
\begin{equation}
P_{it} := \frac{1}{M_{it}} \sum_{k=1}^{M_{it}} \delta_{y_{it}^{(k)}}
\end{equation}
where $\delta$ denotes the Dirac mass. Crucially, we do not require the sample size $M_{it}$ to be balanced. The number of micro-observations can vary arbitrarily across units and time periods, allowing for the flexible integration of unbalanced survey data.

For the post-intervention period $t > T_0$, we observe the treated distribution $P_{1t} = \mathcal{L}(Y_{1t}(1))$, but we seek to estimate the counterfactual distribution $P_{1t}^N = \mathcal{L}(Y_{1t}(0))$. We construct a synthetic control as a convex combination of the donor distributions. Let $\Lambda = \Delta^{J}$ denote the simplex defined as follows:
\begin{equation}
\Lambda = \left\{ \lambda \in \mathbb{R}^{J} \mid \lambda_j \ge 0, \sum_{j=2}^{J+1} \lambda_j = 1 \right\}
\end{equation}
For any weight vector $\lambda \in \Lambda$, we define the synthetic distribution $P_{\lambda, t}$ as the linear mixture of donor measures:
\begin{equation}
P_{\lambda, t} := \sum_{j=2}^{J+1} \lambda_j P_{jt}
\end{equation}
This formulation imposes no parametric constraints on the shape of the counterfactual, allowing for the recovery of complex, multimodal distributions.

\subsection{Estimation and Aggregation}

The core innovation of our framework lies in the choice of the discrepancy metric used to determine the optimal weights. While standard methods typically minimize the integrated squared error between quantile functions, we adopt the Wasserstein-1 distance to ensure robustness against support mismatch.

Formally, utilizing the Kantorovich-Rubinstein duality, we define the transport cost between the treated unit and a synthetic mixture at any specific time period $t$ as:
\begin{equation}
W_{1}(P_{\lambda,t}, P_{1t}) = \sup_{f \in Lip_{1}(\mathcal{X})} \left( \mathbb{E}_{x \sim P_{1t}}[f(x)] - \sum_{j=2}^{J+1} \lambda_j \mathbb{E}_{x \sim P_{jt}}[f(x)] \right)
\end{equation}
where $Lip_{1}(\mathcal{X})$ denotes the set of 1-Lipschitz continuous functions. Unlike Euclidean metrics, this distance provides valid gradient signals even when the supports of the distributions are disjoint, as the gradient is governed by the underlying transport cost rather than pointwise density overlaps.

To construct the final synthetic control, our estimator aggregates information across the pre-intervention window. We adopt a period-wise estimation strategy to capture the structural fit at each time step. For each pre-intervention period $t \le T_0$, we define the period-specific optimal weights $\lambda_{t}^{*}$ as the minimizer of the regularized Wasserstein loss:
\begin{equation}\label{argmin}
\lambda_{t}^{*} = \mathop{\mathrm{argmin}}_{\lambda \in \Lambda} \left( W_{1}(P_{\lambda,t}, P_{1t}) - \eta H(\lambda) \right)
\end{equation}
where $H(\lambda) = - \sum_{j=2}^{J+1} \lambda_j \log \lambda_j$ is the entropic regularization term. The parameter $\eta > 0$ controls the trade-off between fitting the transport cost and maximizing the entropy of the weights, ensuring solution uniqueness and numerical stability. Note that while the objective function varies with $t$ through the observed distributions, the regularization term $H(\lambda)$ depends solely on the weight vector and is time-invariant.

The final synthetic control weights $\lambda^{*}$ are obtained by aggregating these period-specific estimates using temporal weights $\{w_{t}\}_{t=1}^{T_{0}}$ (typically $w_{t} = 1/T_{0}$ for a simple average):
\begin{equation}
\lambda^{*} = \sum_{t=1}^{T_{0}} w_{t} \lambda_{t}^{*}
\end{equation}
This two-step aggregation procedure stabilizes the estimation by reducing the variance associated with any single period's idiosyncratic shocks, while ensuring that the final weights reflect a consistent structural similarity across the entire pre-intervention history.\footnote{We employ a linear aggregation of period-specific weights rather than a single global optimization to accommodate potential structural breaks or heteroskedasticity in the pre-treatment period. However, alternative aggregation schemes, such as weighting periods by their inverse Wasserstein loss, are straightforward extensions of this framework.}

\subsection{Structural Identification}

We now establish the conditions under which minimizing the pre-intervention discrepancy identified in Eq. \ref{argmin} recovers the true counterfactual distribution. Following the literature on synthetic controls, we assume the outcomes are generated by a latent factor structure, which we extend to the space of probability measures.

\begin{assumption}[Latent Factor Generation]\label{latent}
    For each unit $i$ and time $t$, the potential outcome distribution under control, $P_{it}^N$, is generated by a time-varying transport map $\mathcal{T}_{t}: \mathcal{P}(\mathcal{Z}) \to \mathcal{P}(\mathcal{X})$ acting on a unit-specific, time-invariant distribution of latent factors $\mu_{i} \in \mathcal{P}(\mathcal{Z})$:
\begin{equation}
P_{it}^N = \mathcal{T}_{t}(\mu_{i})
\end{equation}
This formulation generalizes the standard interactive fixed effects model, $Y_{it} = \delta_t + \theta_t' Z_i + \epsilon_{it}$, where $\mathcal{T}_t$ represents the aggregate structural evolution of the economy acting on unit-specific heterogeneity $\mu_i$.
\end{assumption}

\begin{assumption}[Structural Stability via Scaled Isometry]
    The family of transport maps $\{\mathcal{T}_{t}\}_{t=1}^{T}$ satisfies a scaled isometry condition with respect to the Wasserstein metric. Specifically, for any two latent measures $\nu_{1}, \nu_{2} \in \mathcal{P}(\mathcal{Z})$ and any time $t$, there exists a scaling factor $\kappa_{t} > 0$ such that:
\begin{equation}
W_{1}(\mathcal{T}_{t}(\nu_{1}), \mathcal{T}_{t}(\nu_{2})) = \kappa_{t} W_{1}(\nu_{1}, \nu_{2})
\end{equation}
\end{assumption}

\textbf{Remark.} Assumption 2 implies that while aggregate shocks may alter the location, scale, or shape of outcomes (scaling distances by $\kappa_t$), they preserve the relative structural geometry between units. This condition serves as a distributional generalization of the parallel trends assumption.

Based on this structural geometry, Theorem 1 establishes that a synthetic control constructed to match the pre-treatment observables identifies the latent counterfactual.

\paragraph{Theorem 1 (Counterfactual Recovery).}
Suppose Assumptions 1 and 2 hold. Let $\lambda^{*} \in \Lambda$ be a weight vector satisfying the following conditions:
\begin{enumerate}
    \item Observable Matching: $\sum_{j=2}^{J+1} \lambda_{j}^{*} P_{jt} = P_{1t}$ for all $t \le T_{0}$.
    \item Latent Convexity: The treated unit's latent factors lie within the convex hull of the donors, i.e., $\mu_1 = \sum_{j=2}^{J+1} \lambda_j^* \mu_j$.
\end{enumerate}
Then, $\lambda^{*}$ perfectly recovers the counterfactual distribution in the post-intervention period:
\begin{equation}
\sum_{j=2}^{J+1} \lambda_{j}^{*} P_{jt} = P_{1t}^{N}, \quad \forall t > T_{0}
\end{equation}

\paragraph{Proof.} See Appendix A.

\subsection{Geometric Robustness and Theoretical Comparison}

A critical theoretical advantage of our framework over standard distributional synthetic controls lies in its robustness to support mismatch and model misspecification. We analyze two distinct geometric failure modes of existing estimators and demonstrate how the Wasserstein metric resolves them.

\paragraph{Gradient Existence under Support Mismatch.}
A key computational advantage of the Wasserstein metric is the preservation of informative gradients even when distributions have disjoint supports. This property explains the stability of our estimator in Simulation 2, where $L_2$-based methods fail. We formalize this property in the following proposition.

\begin{proposition}[Non-vanishing Gradients]
\label{prop:gradients}
Let $P_{1}$ be the treated distribution and $P_{\lambda} = \sum_{j=2}^{J+1} \lambda_j P_{j}$ be the synthetic mixture. Assume the supports are disjoint, i.e., $\text{supp}(P_{1}) \cap \text{supp}(P_{\lambda}) = \emptyset$, and that the distributions are absolutely continuous with densities $p_1$ and $p_{\lambda}$.
\begin{enumerate}
    \item \textbf{$L_2$ Failure:} If the discrepancy is measured by the squared $L_2$ distance between densities, the gradient of the loss with respect to the weights vanishes almost everywhere:
    \begin{equation}
    \nabla_{\lambda} ||p_{1} - p_{\lambda}||_2^2 = 0
    \end{equation}
    \item \textbf{$W_1$ Validity:} If the discrepancy is measured by the Wasserstein-1 distance, the sub-gradient with respect to the weights is non-trivial and is given by:
    \begin{equation}
    \frac{\partial W_{1}(P_{\lambda}, P_{1})}{\partial \lambda_{j}} = -\mathbb{E}_{x \sim P_{j}}[f^{*}(x)]
    \end{equation}
    where $f^{*}$ is the optimal Kantorovich potential.
\end{enumerate}
\end{proposition}

\paragraph{Proof.} See Appendix B.

Beyond support mismatch, our framework also addresses structural rigidity in the shape of the counterfactual. Existing approaches \citep[e.g.,][]{gunsilius2023distributional} typically define the synthetic control as a weighted average of quantile functions: $Q_{synth}(\tau) = \sum \lambda_j Q_{j}(\tau)$. This formulation implicitly assumes that the counterfactual distribution belongs to a location-scale family generated by the donors. Such an assumption imposes a rigid unimodal constraint: an average of unimodal quantile functions typically results in a unimodal distribution.
Our framework, by defining the synthetic control as a linear mixture of measures ($P_{synth} = \sum \lambda_j P_{j}$), relaxes this restriction. As demonstrated in our bimodal simulation (Simulation 3), a mixture of unimodal donors can naturally approximate a multimodal target. Thus, our estimator is robust to misspecification regarding the shape of the counterfactual distribution.

\section{Estimation and Computation}

While the identification results in the previous section rely on population distributions, in practice we observe only finite samples. This section details the estimation procedure, replacing the intractable supremum over Lipschitz functions with a neural network approximation, and outlines the adversarial optimization framework used to solve for the synthetic weights.

\subsection{The Empirical Minimax Formulation}

Let $\{y_{it}^{(k)}\}_{k=1}^{M_{it}}$ denote the observed micro-samples for unit $i$ at time $t$. We approximate the underlying population distribution $P_{it}$ using the empirical measure $\hat{P}_{it}$:
\begin{equation}
\hat{P}_{it} = \frac{1}{M_{it}} \sum_{k=1}^{M_{it}} \delta_{y_{it}^{(k)}}
\end{equation}
where $\delta$ is the Dirac mass.

Directly computing the Wasserstein distance $W_1(\hat{P}_{\lambda, t}, \hat{P}_{1t})$ involves solving a high-dimensional linear programming problem, which becomes computationally prohibitive when the sample size $M_{it}$ is large. Moreover, we require differentiable gradients with respect to the weights $\lambda$ to optimize the synthetic control. To address these challenges, we adopt the WGAN framework \citep{arjovsky2017wasserstein}.

We parameterize the Kantorovich potential, which is the critic function, using a neural network $f_{\theta}: \mathcal{X} \to \mathbb{R}$ with parameters $\theta \in \Theta$, where $\Theta$ is a compact parameter space. The space of 1-Lipschitz functions is thus approximated by the parameterized family $\mathcal{F}_{\Theta} = \{f_{\theta} : \theta \in \Theta\}$.
Consequently, the estimation of the period-specific weights $\lambda_{t}^{*}$ is recast as the solution to the following empirical minimax problem:
\begin{equation}
\min_{\lambda \in \Lambda} \max_{\theta \in \Theta} \mathcal{L}_{t}(\lambda, \theta; \eta)
\label{minmax}
\end{equation}
where the regularized empirical objective function $\mathcal{L}_{t}$ is given by:
\begin{equation}\label{w1}
\mathcal{L}_{t}(\lambda, \theta; \eta) = \left( \mathbb{E}_{x \sim \hat{P}_{1t}}[f_{\theta}(x)] - \sum_{j=2}^{J+1} \lambda_{j} \mathbb{E}_{x \sim \hat{P}_{jt}}[f_{\theta}(x)] \right) - \eta \sum_{j=2}^{J+1} \lambda_{j} \log \lambda_{j}
\end{equation}
In this formulation, the expectations over empirical measures are computed as sample means (e.g., $\mathbb{E}_{x \sim \hat{P}_{jt}}[f_{\theta}(x)] = \frac{1}{M_{jt}} \sum_{k} f_{\theta}(y_{jt}^{(k)})$). The inner maximization over $\theta$ approximates the Wasserstein distance, which is finding the tightest lower bound on the transport cost, while the outer minimization over $\lambda$ finds the optimal weights that minimize this estimated distance plus the entropic penalty.

\subsection{Network Architecture and Gradient Penalty}

Solving the minimax problem in Eq.~\eqref{minmax} requires the critic network $f_{\theta}$ to belong to the set of 1-Lipschitz functions, a necessary condition for the Kantorovich-Rubinstein duality to hold. Standard weight clipping techniques often lead to optimization pathologies. To strictly enforce this constraint, we employ the Gradient Penalty method proposed by \cite{gulrajani2017improved}.

We augment the objective function of the critic with a regularization term that penalizes deviations of the gradient norm from unity. The modified loss for the critic becomes:
\begin{equation}
\mathcal{L}_{critic}(\theta) = \left( \mathbb{E}_{x \sim \hat{P}_{1t}}[f_{\theta}(x)] - \sum_{j=2}^{J+1} \lambda_{j} \mathbb{E}_{x \sim \hat{P}_{jt}}[f_{\theta}(x)] \right) - \zeta \mathbb{E}_{\hat{x} \sim \hat{P}_{int}}\left[ (||\nabla_{\hat{x}} f_{\theta}(\hat{x})||_{2} - 1)^{2} \right]
\end{equation}
where $\zeta > 0$ is the penalty coefficient, and $\hat{P}_{int}$ denotes the distribution of random interpolates sampled uniformly along straight lines connecting points from the treated and synthetic distributions. This soft constraint ensures that $f_{\theta}$ approximates the Wasserstein distance with valid gradients almost everywhere.

\subsection{Optimization Algorithm}

We solve the saddle-point problem using an alternating Gradient Descent-Ascent (GDA) procedure. The optimization involves two distinct update rules tailored to the geometry of the parameters.

For a fixed weight vector $\lambda$, we update the network parameters $\theta$ to maximize the separation between the treated and synthetic distributions. We apply standard gradient ascent or Adam on the penalized objective $\mathcal{L}_{critic}$.

For a fixed critic $f_{\theta}$, the minimization with respect to $\lambda$ is strictly convex due to the entropic regularization. However, standard gradient descent is inefficient and may violate the simplex constraint $\sum \lambda_{j}=1$. Instead, we employ Exponentiated Gradient Descent, which naturally preserves the simplex geometry and ensures strictly positive weights.
The multiplicative update rule for each donor $j \in \{2, \dots, J+1\}$ is derived from the first-order optimality condition:
\begin{equation}
\lambda_{j}^{(k+1)} = \frac{\lambda_{j}^{(k)} \exp\left( -\alpha_{\lambda} \left( \mathbb{E}_{x \sim \hat{P}_{jt}}[f_{\theta}(x)] + \eta (1 + \log \lambda_{j}^{(k)}) \right) \right)}{Z^{(k+1)}}
\end{equation}
where $\alpha_{\lambda}$ is the learning rate and $Z^{(k+1)}$ is the normalization factor ensuring the weights sum to one.

The complete estimation procedure, incorporating the period-wise strategy and temporal aggregation, is detailed in Algorithm 1.

\begin{algorithm}[H]
\caption{Robust Distributional Synthetic Controls via WGAN (DSC-WGAN)}
\label{alg:DSC-WGAN}
\begin{algorithmic}[1]

\renewcommand{\algorithmicrequire}{\textbf{Input:}}
\renewcommand{\algorithmicensure}{\textbf{Output:}}

\REQUIRE Micro-samples $\{Y_{it}\}$ for units $i=1,\dots,J+1$ and periods $t=1,\dots,T_{0}$.
\REQUIRE Hyperparameters: Learning rates $\alpha_{\theta}, \alpha_{\lambda}$, penalty $\zeta$, regularization $\eta$, iterations $K_{critic}$.

\State \textbf{Initialize:} $\Lambda_{history} \leftarrow \emptyset$

\For{each pre-treatment period $t = 1, \dots, T_{0}$}
    \State Initialize weights $\lambda_{t} \leftarrow (1/J, \dots, 1/J)$
    \State Initialize critic network parameters $\theta$
    \While{not converged}
        \State \COMMENT{\textit{Step 1: Update Critic (Dual)}}
        \For{$k = 1, \dots, K_{critic}$}
            \State Sample batches $x_{1} \sim \hat{P}_{1t}$ and $x_{synth} \sim \sum \lambda_{j,t} \hat{P}_{jt}$
            \State Sample interpolates $\hat{x}$ and compute Gradient Penalty
            \State $\theta \leftarrow \theta + \alpha_{\theta} \nabla_{\theta} \mathcal{L}_{critic}$
        \EndFor
        \State \COMMENT{\textit{Step 2: Update Weights (Primal)}}
        \State Sample batches from donor units $\{ \hat{P}_{jt} \}_{j=2}^{J+1}$
        \State Compute gradient proxy: $g_{j} = \mathbb{E}[f_{\theta}(Y_{jt})] + \eta(1 + \log \lambda_{j,t})$
        \State Update $\lambda_{t}$ via Mirror Descent (Eq. 14) and normalize
    \EndWhile
    \State Store optimal period weights: $\Lambda_{history} \leftarrow \Lambda_{history} \cup \{ \lambda_{t} \}$
\EndFor

\State \textbf{Aggregation:} Compute final weights $\lambda^{*} = \frac{1}{T_{0}} \sum_{\lambda \in \Lambda_{history}} \lambda$
\Return $\lambda^{*}$
\end{algorithmic}
\end{algorithm}

\subsection{Computation}

The estimation of the optimal weights $\lambda^*$ involves solving a minimax optimization problem. Based on the identification results in Equation (\ref{eq:obj_pop}), the estimator is the solution to:
\begin{equation} \label{eq:minimax}
    \min_{\lambda \in \Delta^{J-1}} \max_{f \in Lip_1(\mathcal{X})} \mathcal{L}(\lambda, f) + \eta \sum_{j=2}^{J+1} \lambda_j \log \lambda_j,
\end{equation}
where $\mathcal{L}(\lambda, f) = \mathbb{E}_{P_1}[f] - \sum_{j=2}^{J+1} \lambda_j \mathbb{E}_{P_j}[f]$.

We propose a stochastic approximation algorithm that alternates between optimizing the critic function $f$ and the synthetic weights $\lambda$.

We parameterize the Kantorovich potential $f$ using a neural network $f_\theta: \mathcal{X} \to \mathbb{R}$ with parameters $\theta \in \Theta$. The network architecture consists of a Multi-Layer Perceptron (MLP) with Leaky ReLU activation functions to ensure efficient gradient flow.

To enforce the 1-Lipschitz constraint required by the Wasserstein duality, we adopt the Gradient Penalty (GP) method proposed by \cite{gulrajani2017improved}. Unlike weight clipping, which biases the critic towards simple functions, GP enforces the Lipschitz constraint softly by penalizing the norm of the gradient. The empirical objective for the critic becomes:
\begin{equation}
    \max_{\theta} \mathcal{J}_{critic}(\theta) = \underbrace{\frac{1}{N}\sum_{k=1}^N f_\theta(Y_{1}^{(k)}) - \sum_{j=2}^{J+1} \lambda_j \left( \frac{1}{N}\sum_{k=1}^N f_\theta(Y_{j}^{(k)}) \right)}_{\text{Transport Cost}} - \underbrace{\gamma_{GP} \cdot \mathbb{E}_{\hat{x} \sim P_{\hat{x}}} [(\|\nabla_{\hat{x}} f_\theta(\hat{x})\|_2 - 1)^2]}_{\text{Gradient Penalty}},
\end{equation}
where $P_{\hat{x}}$ is defined by sampling uniformly along straight lines between pairs of points from the target and donor distributions, and $\gamma_{GP} > 0$ is the penalty coefficient\footnote{We conduct a comprehensive sensitivity analysis in Appendix D, demonstrating that the estimator's performance remains robust to variations in key hyperparameters, including the gradient penalty coefficient and network width.}.

For the outer minimization problem with respect to $\lambda$, we must handle the simplex constraint $\lambda \in \Delta^{J-1}$ and the entropic regularization term. Standard Projected Gradient Descent is suboptimal here as it ignores the geometry induced by the entropy.

Instead, we employ \textit{Entropic Mirror Descent} (also known as the Exponentiated Gradient algorithm), which is the natural first-order method for entropy-regularized problems. The update rule is multiplicative, ensuring that $\lambda$ strictly remains in the positive simplex without explicit projection steps.

At iteration $t$, given the current critic $f_{\theta^{(t)}}$, the gradient of the transport cost with respect to $\lambda_j$ is simply the negative expected value of the critic for donor $j$:
\begin{equation}
    g_j^{(t)} = -\frac{1}{N} \sum_{k=1}^N f_{\theta^{(t)}}(Y_{j}^{(k)}).
\end{equation}
The weight update follows the mirror descent step with learning rate $\alpha_\lambda$:
\begin{equation}
    \lambda_j^{(t+1)} = \frac{\lambda_j^{(t)} \exp\left( - \alpha_\lambda (g_j^{(t)} + \eta (1 + \log \lambda_j^{(t)})) \right)}{\sum_{l=2}^{J+1} \lambda_l^{(t)} \exp\left( - \alpha_\lambda (g_l^{(t)} + \eta (1 + \log \lambda_l^{(t)})) \right)}.
\end{equation}
In practice, since the regularization term $\eta \lambda \log \lambda$ is strictly convex, we can also decouple the update: update $\lambda$ based on the transport gradient $g^{(t)}$ and then apply the proximal operator for the entropy, which simplifies to the softmax-like scaling.

The complete training procedure is detailed in Algorithm \ref{alg:wgan}. Since the sample size $N$ in synthetic control applications is typically moderate (e.g., hundreds of time periods), we can often use full-batch gradient descent rather than mini-batches to reduce variance.

\begin{algorithm}[H]
\caption{Entropic WGAN for Synthetic Controls}
\label{alg:wgan}
\begin{algorithmic}[1]
    \Require Dataset $\{Y_1, Y_2, \dots, Y_{J+1}\}$, regularization $\eta$, gradient penalty $\gamma_{GP}$.
    \Require Learning rates $\alpha_\theta, \alpha_\lambda$, number of critic steps $n_{critic}$.
    \State Initialize weights $\lambda = (1/J, \dots, 1/J)$ and network parameters $\theta$.
    \While{not converged}
        \For{$t = 1$ to $n_{critic}$}
            \State Sample interpolation points $\hat{x}$ for gradient penalty.
            \State Compute gradient $\nabla_\theta \mathcal{J}_{critic}$ (Equation 32).
            \State Update critic: $\theta \leftarrow \theta + \alpha_\theta \text{Adam}(\nabla_\theta \mathcal{J}_{critic})$.
        \EndFor
        \State Compute transport gradients for donors: $g_j = - \frac{1}{N} \sum_{k} f_\theta(Y_j^{(k)})$.
        \State Update weights via Mirror Descent:
        \State $\tilde{\lambda}_j \leftarrow \lambda_j \cdot \exp(-\alpha_\lambda g_j)$
        \State Normalize: $\lambda_j \leftarrow \tilde{\lambda}_j / \sum_l \tilde{\lambda}_l$
        \State \Comment{Implicitly handling regularization via the decay or proximal step}
    \EndWhile
    \Return Optimized weights $\lambda^*$.
\end{algorithmic}
\end{algorithm}

\section{Mathematical Properties}

\subsection{Setup and Identification}

In this section, we establish the statistical identification of the synthetic control weights. While Section 3 establishes the causal identification, for example, the recovery of the unobserved counterfactual distribution under structural assumptions, this section focuses on the statistical identification of the weight vector $\lambda^*$ itself. Specifically, we establish conditions under which the regularized minimization problem admits a unique solution, ensuring that the estimator is well-defined.

Let $(\mathcal{X}, d)$ be a complete separable metric space. Let $\mathcal{M}(\mathcal{X})$ denote the vector space of finite signed Borel measures on $\mathcal{X}$, and let $\mathcal{P}_1(\mathcal{X}) \subset \mathcal{M}(\mathcal{X})$ denote the space of Borel probability measures with finite first moments. We endow $\mathcal{P}_1(\mathcal{X})$ with the 1-Wasserstein distance $W_1$.

We observe the target distribution $P_1 \in \mathcal{P}_1(\mathcal{X})$ and a set of donor distributions $\mathcal{D} = \{P_2, \dots, P_{J+1}\} \subset \mathcal{P}_1(\mathcal{X})$. The synthetic control is defined as a linear mixture $P_\lambda = \sum_{j=2}^{J+1} \lambda_j P_j$, where the weights $\lambda$ reside in the simplex:
\begin{equation}
    \Delta^{J-1} = \left\{ \lambda \in \mathbb{R}^J \mid \lambda_j \ge 0, \sum_{j=2}^{J+1} \lambda_j = 1 \right\}.
\end{equation}

The estimator $\lambda^*$ is defined as the minimizer of the regularized transport cost:
\begin{equation} \label{eq:obj_pop}
    \lambda^* = \arg\min_{\lambda \in \Delta^{J-1}} Q_\eta(\lambda), \quad \text{where } Q_\eta(\lambda) := W_1(P_\lambda, P_1) + \eta \sum_{j=2}^{J+1} \lambda_j \log \lambda_j.
\end{equation}
To ensure the mapping from weights to distributions is injective, we impose the following condition:

\begin{assumption}[Affine Independence of Donors] \label{ass:affine}
    The donor distributions $\{P_2, \dots, P_{J+1}\}$ are affinely independent in $\mathcal{M}(\mathcal{X})$. Specifically,
    \begin{equation}
        \sum_{j=2}^{J+1} \alpha_j P_j = 0 \quad \text{and} \quad \sum_{j=2}^{J+1} \alpha_j = 0 \implies \alpha_j = 0, \quad \forall j.
    \end{equation}
\end{assumption}

\begin{theorem}[Global Identification] \label{thm:identification}
    Suppose Assumption \ref{ass:affine} holds and $\eta > 0$. Then, the regularized objective function $Q_\eta(\lambda)$ is strictly convex on $\Delta^{J-1}$, and there exists a unique global minimizer $\lambda^* \in \Delta^{J-1}$.
\end{theorem}

\begin{proof}
    See Appendix C.1.
\end{proof}

\subsection{Consistency}

We examine the asymptotic properties of the estimator in the regime of Large Micro-Sample, Fixed Donors. Specifically, we consider the case where the number of micro-level observations per period, denoted by $N$ (or $M_{it}$), tends to infinity, while the number of donor units $J$ remains fixed. This regime is relevant for distributional synthetic controls where researchers have access to granular data (e.g., census or survey micro-data) for each aggregate unit.

Let $\{Y_{1}^{(k)}\}_{k=1}^N$ be $N$ i.i.d. observations from the target $P_1$, and $\{Y_{j}^{(k)}\}_{k=1}^N$ be $N$ i.i.d. observations from donor $P_j$ for $j \in \{2, \dots, J+1\}$. The empirical measures are $\hat{P}_{1,N} = \frac{1}{N} \sum_{k=1}^N \delta_{Y_{1}^{(k)}}$ and $\hat{P}_{j,N} = \frac{1}{N} \sum_{k=1}^N \delta_{Y_{j}^{(k)}}$.

The estimator minimizes the empirical regularized objective:
\begin{equation}
    \hat{\lambda}_N = \arg\min_{\lambda \in \Delta^{J-1}} \hat{Q}_N(\lambda), \quad \text{where } \hat{Q}_N(\lambda) := W_1\left(\sum_{j=2}^{J+1} \lambda_j \hat{P}_{j,N}, \hat{P}_{1,N}\right) + \eta \Omega(\lambda).
\end{equation}

\begin{assumption}[Finite Moments] \label{ass:moments}
    For all $j \in \{1, \dots, J+1\}$, $\mathbb{E}_{P_j}[d(Y, y_0)] < \infty$ for some $y_0 \in \mathcal{X}$.
\end{assumption}

\begin{theorem}[Consistency] \label{thm:consistency}
    Under Assumptions \ref{ass:affine} and \ref{ass:moments}, and $\eta > 0$, $\hat{\lambda}_N \xrightarrow{p} \lambda^*$ as $N \to \infty$.
\end{theorem}

\begin{proof}
    See Appendix C.2.
\end{proof}

\subsection{Asymptotic Normality}

To derive the limiting distribution of $\hat{\lambda}_N$, we utilize the theory of Sieve M-estimation (Chen, 2007). We treat the critic network $f_\theta$ as a sieve approximator for the dual potential space $Lip_1(\mathcal{X})$.

Let $\mathcal{F}_N = \{f_\theta : \theta \in \Theta_N\}$ be the class of functions parameterized by the neural network, where the complexity of the network (e.g., width and depth) may grow with $N$. Let $\lambda^*$ be the population minimizer and $f^*(\cdot, \lambda)$ be the optimal Kantorovich potential for a given $\lambda$.

We impose the following regularity conditions:

\begin{assumption}[Local Curvature] \label{ass:hessian}
    The population objective $Q_\eta(\lambda)$ is twice continuously differentiable in a neighborhood of $\lambda^*$, and the Hessian matrix $H^* = \nabla_{\lambda}^2 Q_\eta(\lambda^*)$ is positive definite.
\end{assumption}

\begin{assumption}[Sieve Approximation and Complexity]
\label{ass:sieve_rate}
Let $\mathcal{F}_N$ be the class of functions implemented by the neural network with width $W_N$ and depth $L_N$. Let $f^*(\cdot, \lambda)$ denote the true Kantorovich potential. We assume that the approximation bias $\delta_N = \sup_{\lambda \in \Delta^{J-1}} \inf_{f \in \mathcal{F}_N} \| f - f^*(\cdot, \lambda) \|_{\infty}$ satisfies the rate condition:
\begin{equation}
    \sqrt{N} \delta_N \to 0 \quad \text{as} \quad N \to \infty.
\end{equation}
Furthermore, the complexity of the sieve, measured by the covering number $\log \mathcal{N}(\epsilon, \mathcal{F}_N, \|\cdot\|_\infty)$, grows sufficiently slowly relative to $N$ such that the stochastic equicontinuity condition (Assumption 7) holds.
\end{assumption}

\begin{remark}[Overcoming the Curse of Dimensionality]
Assumption \ref{ass:sieve_rate} imposes a restriction on the approximation error of the neural network. In standard nonparametric sieve estimation, achieving a rate faster than $N^{-1/2}$ is typically prohibited by the curse of dimensionality when the dimension of the outcome $d$ is large (i.e., the rate scales as $N^{-1/d}$). However, recent theoretical advances in deep learning econometrics suggest that deep neural networks can overcome this limitation if the underlying target function possesses a compositional structure. specifically, Farrell, Liang, and Misra (2021) demonstrate that if the true function (here, the Kantorovich potential $f^*$) belongs to a composition of Hölder smooth functions, deep ReLU networks achieve approximation rates that depend only on the intrinsic dimension of the composition, not the ambient dimension $d$. Given that the latent factor structure in synthetic controls (Assumption 1) implies a low-dimensional generative process, the requirement that $\sqrt{N}\delta_N \to 0$ is theoretically justified for sufficiently deep architectures.
\end{remark}

\begin{assumption}[Stochastic Equicontinuity and Entropy]
\label{ass:equicontinuity}
The class of functions $\mathcal{F}_N$ satisfies the bracketing entropy condition required for the Functional Central Limit Theorem. Specifically, let $J_{[]}(\delta, \mathcal{F}_N, \|\cdot\|)$ denote the bracketing entropy integral. We assume that the complexity of the neural network grows sufficiently slowly such that:
\begin{equation}
    \int_0^1 \sqrt{\log \mathcal{N}_{[]}(\epsilon, \mathcal{F}_N, \|\cdot\|_{2,P})} d\epsilon < \infty \quad \text{and} \quad \frac{1}{\sqrt{N}} J_{[]}(\delta_N, \mathcal{F}_N, \|\cdot\|) \to 0,
\end{equation}
where $\mathcal{N}_{[]}$ is the covering number. This ensures that the empirical process $\mathbb{G}_N f = \sqrt{N}(\mathbb{P}_N - P)f$ is asymptotically equicontinuous over the sieve space.
\end{assumption}

\begin{theorem}[Asymptotic Normality] \label{thm:normality}
    Under Assumptions \ref{ass:affine}-\ref{ass:equicontinuity}, the estimator satisfies:
    \begin{equation}
        \sqrt{N}(\hat{\lambda}_N - \lambda^*) \xrightarrow{d} \mathcal{N}(0, \Sigma),
    \end{equation}
    where $\Sigma = (H^*)^{-1} V (H^*)^{-1}$, and $V = \text{Var}(S(Y, \lambda^*))$ is the variance of the efficient score function defined in Appendix C.
\end{theorem}

\begin{proof}
    See Appendix C.3.
\end{proof}

\begin{remark}[Curse of Dimensionality]
    Assumption \ref{ass:sieve_rate} requires the approximation bias to vanish faster than the parametric rate $N^{-1/2}$. In high-dimensional settings ($d \gg 1$), this condition implies that the underlying Kantorovich potentials must belong to a smoothness class (e.g., Barron space) efficiently approximable by neural networks. If this condition fails, $\hat{\lambda}_N$ remains a consistent estimator of the ``Neural SCM'' weights—parameters that minimize the transport cost relative to the discriminator class $\mathcal{F}_N$ rather than the exact $Lip_1(\mathcal{X})$ space.
\end{remark}

\subsection{Finite-Sample Inference via Permutation}
While Theorem \ref{thm:normality} establishes that the WGAN-based estimator achieves the parametric $\sqrt{N}$ convergence rate asymptotically, relying on the Gaussian approximation for hypothesis testing presents practical challenges in the context of Synthetic Controls.

First, the number of donor units $J$ is typically small, rendering the estimation of the asymptotic variance matrix $\Sigma = (H^*)^{-1} V (H^*)^{-1}$ unstable. Second, and more critically, the optimal weights $\lambda^*$ often lie on the boundary of the simplex $\Delta^{J-1}$ (i.e., sparsity). In such cases, the finite-sample distribution of the estimator is often non-standard and bounded, deviating significantly from the unrestricted Gaussian limit derived under interior solutions.

Therefore, for valid finite-sample inference, we adopt the conformal-style permutation framework \citep{abadie2010synthetic, firpo2009unconditional} , which remains exact regardless of the sample size or the boundary nature of the solution, provided the exchangeability assumption holds.

\subsubsection{Permutation Test Procedure}

We test the null hypothesis $H_0$ that the intervention has no effect on the distribution of the outcome. Under $H_0$, the treated unit and the donor units are exchangeable regarding the assignment of the intervention. We construct a reference distribution of the treatment effect statistic by iteratively assigning the treatment status to each donor unit $j \in \{2, \dots, J+1\}$.

Let $\hat{\tau}_1$ be the estimated treatment effect for the actual treated unit, defined as the Wasserstein distance between the observed post-treatment distribution and its synthetic counterpart:
\begin{equation}
    \hat{\tau}_1 = W_1(\hat{P}_{1}^{\text{post}}, \hat{P}_{1(\text{synth})}^{\text{post}}).
\end{equation}
Similarly, for each donor $j$, we estimate a placebo synthetic control using the remaining units and compute the placebo effect $\hat{\tau}_j$.

The p-value for the null hypothesis of no effect is calculated as the proportion of units  whose estimated effect is at least as large as that of the treated unit:
\begin{equation}
    p\text{-value} = \frac{1}{J+1} \sum_{j=1}^{J+1} \mathbb{I}(\hat{\tau}_j \ge \hat{\tau}_1),
\end{equation}
where $\mathbb{I}(\cdot)$ is the indicator function.

\subsubsection{Validity of the Permutation Test}

The validity of this p-value relies on the assumption of random assignment of the treatment among the units.

\begin{theorem}[Validity of Permutation Inference] \label{thm:permutation}
    Suppose that under the null hypothesis $H_0$, the assignment of the treatment is uniformly distributed over the set of units $\{1, \dots, J+1\}$ (Exchangeability). Then, the p-value defined in Equation (30) is valid, meaning that for any significance level $\alpha \in [0, 1]$,
    \begin{equation}
        \mathbb{P}_{H_0}(p\text{-value} \le \alpha) \le \alpha.
    \end{equation}
\end{theorem}

\begin{proof}
    See Appendix C.4.
\end{proof}

\begin{remark}[Advantages over Asymptotic Approximation]
    Unlike asymptotic inference which relies on large $N$ approximations (Theorem \ref{thm:normality}), the permutation test is exact in finite samples under the exchangeability assumption. This is particularly advantageous in synthetic control applications where the number of donors $J$ is typically small, rendering large-sample approximations of the weight variance unreliable.
\end{remark}

\section{Simulation Results}

\subsection{Finite-Sample Robustness to Heavy-Tailed Contamination}
\label{subsec:outliers}

This subsection evaluates the finite-sample stability of the proposed estimator under heavy-tailed data contamination. In empirical economics, particularly within finance and income distribution studies, observed outcomes frequently exhibit fat tails or measurement errors that violate sub-Gaussian assumptions. We test the ability of the Wasserstein-based solver to recover the true structural weights $\bm{\lambda}^{true}$ in a regime where the treated unit's post-treatment distribution follows a Huber contamination model.

\subsubsection{Data Generating Process and Contamination Mechanism}

We simulate a panel comprising $J=4$ donor units and one treated unit. To align with the asymptotic regime defined in Section 6.2 ("Large Micro-Sample"), we set the number of micro-observations $N=300$ for each unit $j$ and period $t$. The pre-intervention window is set to $T_0=3$ to stress-test the estimator under limited temporal information.

The outcomes are generated via a latent factor structure consistent with Assumption \ref{latent}:
\begin{equation}
    Y_{j,t,i} = \alpha_t + \mu_j F_t + \sigma \epsilon_{j,t,i}, \quad \epsilon_{j,t,i} \sim \mathcal{N}(0,1)
\end{equation}
where $\alpha_t$ is a common time trend, $F_t \sim \mathcal{N}(0,1)$ represents a time-varying common factor, and $\mu_j \sim \mathcal{U}[0,1]$ is the unit-specific factor loading (corresponding to the latent factors $\mu_i$ in Section 3). The clean counterfactual outcome is constructed as a linear mixture of the donors using the ground-truth weight vector $\bm{\lambda}^{true} = (0.15, 0.25, 0.35, 0.25)^\top$, such that $\lambda_{1} = \sum_{j=2}^{J+1} \lambda_j^{true} \lambda_j$.

To induce contamination, we model the observed treated sample $\{Y_{t,i}^{Treated}\}$ as a mixture distribution:
\begin{equation}
    P_{obs} = (1 - \epsilon) P_{clean} + \epsilon P_{outlier}
\end{equation}
where $\epsilon \in \{0.01, \dots, 0.04\}$ represents the contamination rate. The contaminating distribution $P_{outlier}$ is drawn from a high-variance normal distribution $\mathcal{N}(\mu_{out}, 10^2)$, with the location parameter $\mu_{out}$ shifted significantly into the right tail of the clean distribution.

We assess the performance of two estimators across $N_{SIM}=100$ Monte Carlo replications. The first is the Benchmark (CDF-$L_2$), serving as a proxy for standard distributional synthetic controls (e.g., Gunsilius, 2023), which minimizes the integrated squared error between empirical CDFs. The second is our Proposed (WGAN-GP) estimator, which minimizes the Wasserstein-1 distance via the Kantorovich-Rubinstein dual, parameterized by a 1-Lipschitz neural network.

\subsubsection{Results: Geometric vs. Vertical Sensitivity}

Figure \ref{fig:dgp1} illustrates the geometry of the contamination for period $t=2$. The outliers are strictly localized in the right tail. Under the Wasserstein metric, this shift results in a linear increase in transport cost, whereas under Euclidean metrics, it induces a global distortion of the cumulative probability mass.

\begin{figure}[H]
    \centering
    \includegraphics[width=1\linewidth]{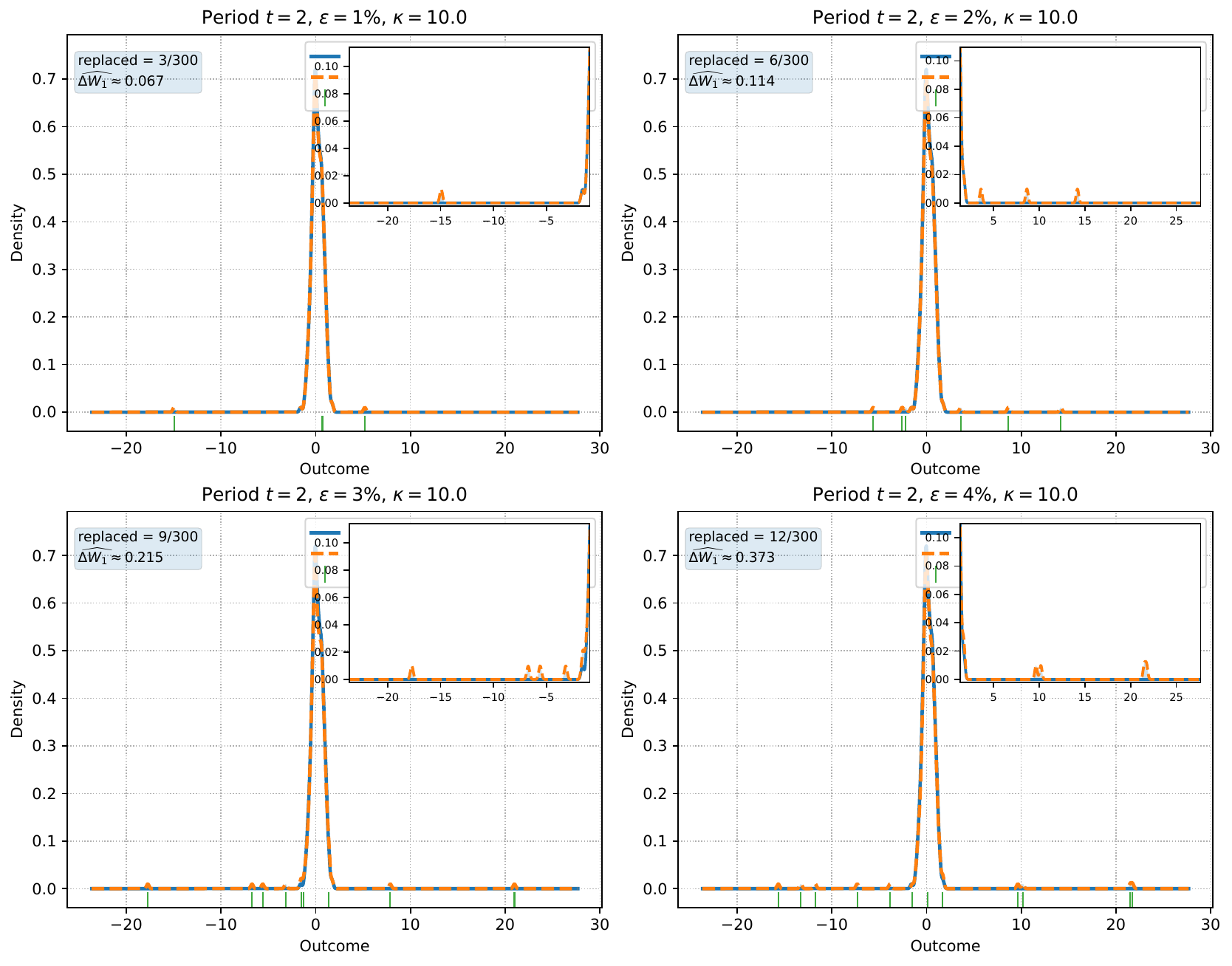}
    \caption{\textbf{Visual Comparison of Clean and Contaminated Target Distributions.} The panels depict the density of the treated outcome in period $t=2$ ($N=300$) under varying degrees of heavy-tailed contamination. The solid blue line represents the clean, ground-truth distribution, while the dashed orange line shows the observed contaminated distribution. As $\epsilon$ increases from $1\%$ to $4\%$, the Wasserstein-1 distance between the clean and contaminated samples ($\widehat{\Delta W}_1$) rises monotonically. The outliers are strictly localized in the tail, leaving the central mode structurally intact.}
    \label{fig:dgp1}
\end{figure}

We evaluate estimator performance using the Average Root Mean Squared Error (RMSE) of the estimated weights relative to $\bm{\lambda}^{true}$. Figure \ref{fig:robustness} presents the results.

\begin{figure}[H]
    \centering
     \includegraphics[width=1\textwidth]{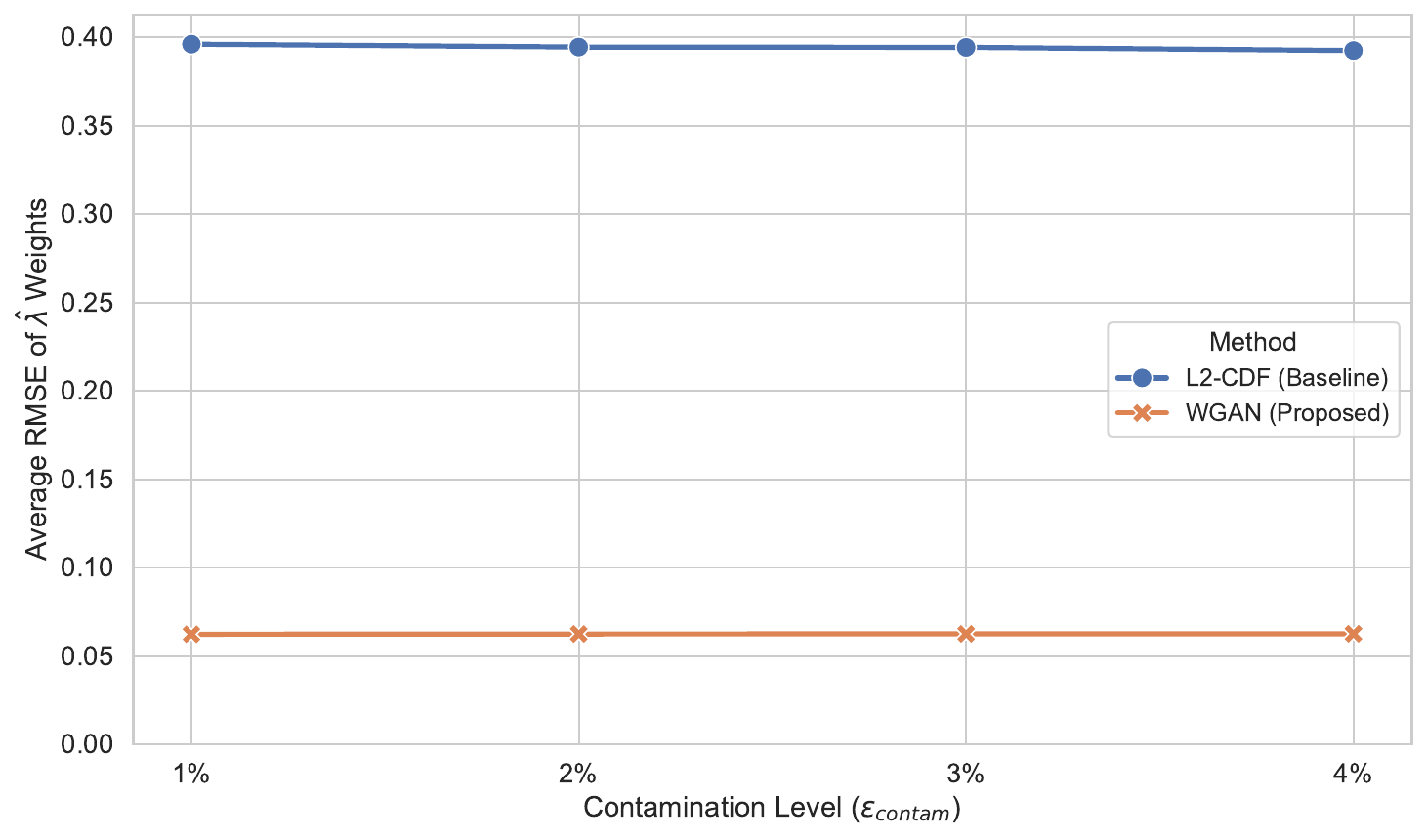}
    \caption{\textbf{Estimator Robustness to Heavy-Tailed Contamination.} The plot displays the Average RMSE of the estimated weights ($\hat{\bm{\lambda}}$) over $N_{SIM}=100$ simulations. The Baseline $L_2$ estimator (blue) exhibits high variance and bias immediately upon contamination ($\text{RMSE} \approx 0.39$), whereas the proposed WGAN estimator (orange) maintains structural fidelity ($\text{RMSE} \approx 0.06$) across all levels.}
    \label{fig:robustness}
\end{figure}

The results reveal a fundamental difference in the robustness of the two metrics. The Benchmark $CDF-L_2$ estimator exhibits an immediate breakdown, with the RMSE stabilizing at a high level ($\approx 0.39$) even under minimal contamination ($1\%$). This fragility is intrinsic to the geometry of the $L_2$ norm on function spaces. An outlier in the far tail shifts the empirical CDF ($F_N(y)$) by a constant amount $1/N$ over the entire domain between the bulk and the outlier. The squared error integral accumulates this deviation over the potentially infinite gap, rendering the loss function highly sensitive to the location of extreme values.

In contrast, the proposed WGAN estimator maintains a remarkably stable RMSE ($\approx 0.06$). By constraining the critic function $f$ to be 1-Lipschitz , the Wasserstein objective ensures that the gradient contribution of any data point is bounded. Consequently, the cost of fitting distant outliers grows only linearly with their distance, effectively "capping" their influence. The optimal transport plan thus prioritizes matching the structural core of the distribution  rather than distorting the weights to chase tail events.

\subsection{Stability under Support Violation}
\label{sec:sim_support_mismatch}

This subsection evaluates estimator performance under support mismatch, a scenario where the support of the target distribution is disjoint from the convex hull of the donor supports. This violation of the classical overlap assumption renders the synthetic control problem theoretically ill-posed. In such regimes, the objective is no longer to achieve perfect recovery , but to test whether the estimator degrades predictability or signals uncertainty through stable, maximum-entropy weights.

\subsubsection{Data Generating Process: Disjoint Supports}

We modify the latent factor structure to induce a controllable geometric gap. Let $Z \sim \mathcal{N}(0,1)$ be a latent generator. The treated unit's factors are generated via a bounded transformation $U_t^{Treated} = 0.5 \tanh(Z)$, confining the target support to the open interval $\mathcal{S}_{1} = (-0.5, 0.5)$.

The donor factors are generated by projecting the same latent $Z$ onto a disjoint domain controlled by a gap parameter $\gamma \in [0, 0.9]$:
\begin{equation}
    U_{j,t} = \text{sign}(Z) \cdot \max(|Z|, \gamma).
\end{equation}
The effective support of the donor pool is $\mathcal{S}_{donors} \approx \mathbb{R} \setminus (-\gamma, \gamma)$. As $\gamma$ increases, the distance between the closure of the treated support and the donor support, $d(\bar{\mathcal{S}}_1, \bar{\mathcal{S}}_{donors})$, strictly increases. At $\gamma=0.9$, the distributions are mutually singular with a significant gap.

To handle the potential non-uniqueness of weights in disjoint regimes, we employ the Entropy-Regularized WGAN (WGAN-E) estimator. The regularization term serves as a strict convexifier, selecting the maximum entropy solution among minimizing sequences:
\begin{equation}
    \hat{\bm{\lambda}}^{WGAN-E} = \arg\min_{\bm{\lambda} \in \Delta^{J-1}} \left[ W_1\left( \sum_{j=1}^J \lambda_j \hat{P}_{j,t}, \hat{P}_{t}^{Treated} \right) + \eta \sum_{j=1}^J \lambda_j \log \lambda_j \right].
\end{equation}
We set $\eta=0.01$ to prioritize transport fidelity while ensuring numerical stability when the transport gradient becomes flat. The benchmark remains the standard CDF-$L_2$ estimator.

\subsubsection{Results: Vanishing Gradients vs. Informative Transport}

We analyze the variance of the estimated weights as a function of the mismatch degree $\gamma$. Figure \ref{fig:support_mismatch} plots the average variance of $\hat{\bm{\lambda}}$ ($N_{SIM}=100$) on a logarithmic scale.

\begin{figure}[htbp]
    \centering
    \includegraphics[width=1\textwidth]{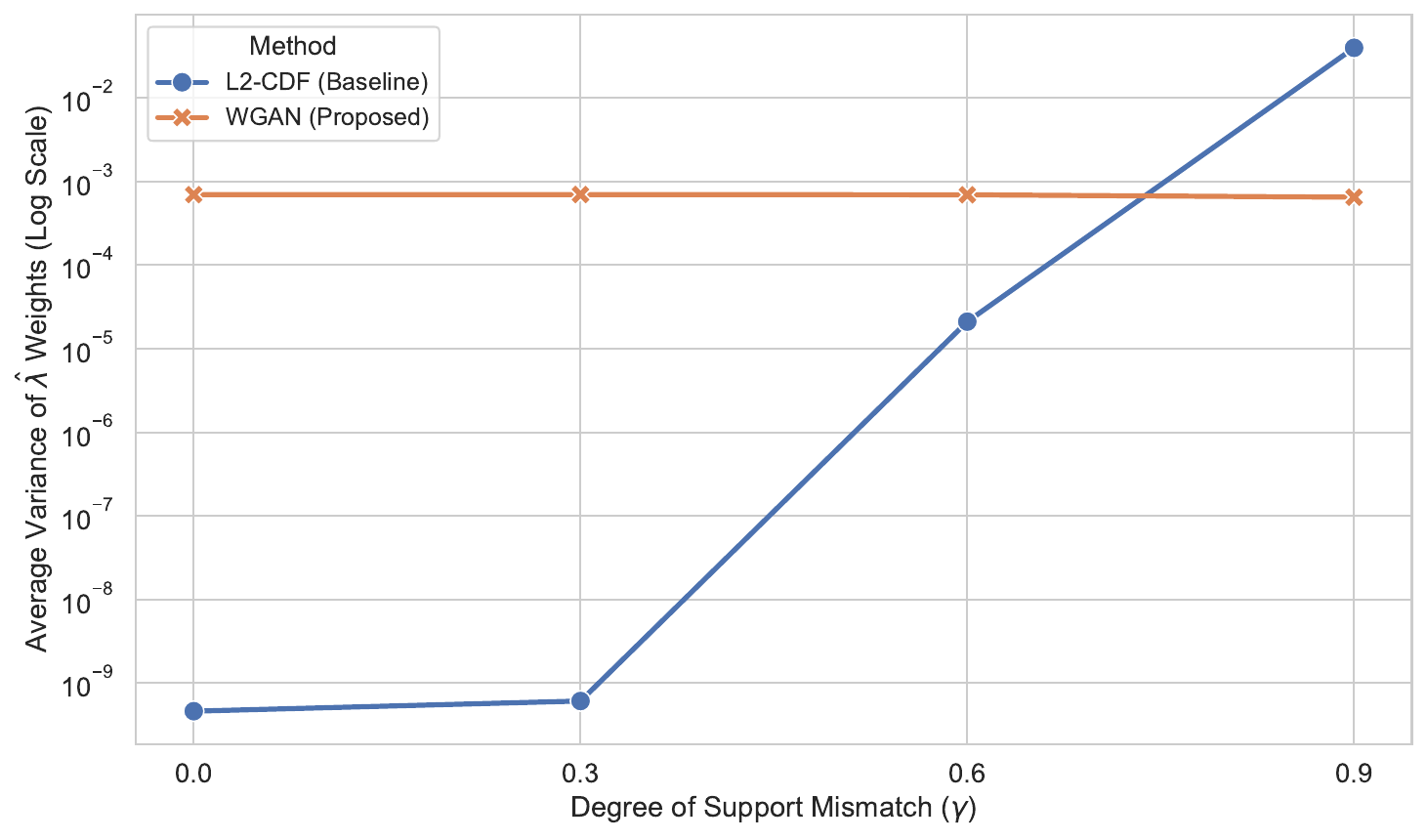}
    \caption{\textbf{Estimator Variance vs. Degree of Support Mismatch ($\gamma$).} The plot displays the Average Variance of the estimated weights. The $L_2$-CDF estimator (blue) exhibits asymptotic instability as the supports become disjoint ($\gamma \to 0.9$). The WGAN-E estimator (orange) maintains bounded variance across the entire domain.}
    \label{fig:support_mismatch}
\end{figure}

The results demonstrate a topological divergence. The Benchmark $L_2$ estimator exhibits instability as the support gap widens, with variance increasing by orders of magnitude at $\gamma=0.9$. This failure is geometric: when $\text{supp}(P_{synth}) \cap \text{supp}(P_{target}) = \emptyset$, the $L_2$ distance between densities (or CDFs) becomes locally constant with respect to infinitesimal weight changes. Consequently, $\nabla_{\bm{\lambda}} \mathcal{L}_{L_2} \approx \mathbf{0}$, causing the numerical optimizer to oscillate between corner solutions driven by machine precision noise.

In contrast, the WGAN-E estimator remains stable. To elucidate the mechanism, Figure \ref{fig:mean_behavior} plots the mean estimated weights $E[\hat{\lambda}_j]$.

\begin{figure}[htbp]
    \centering
    \includegraphics[width=1\textwidth]{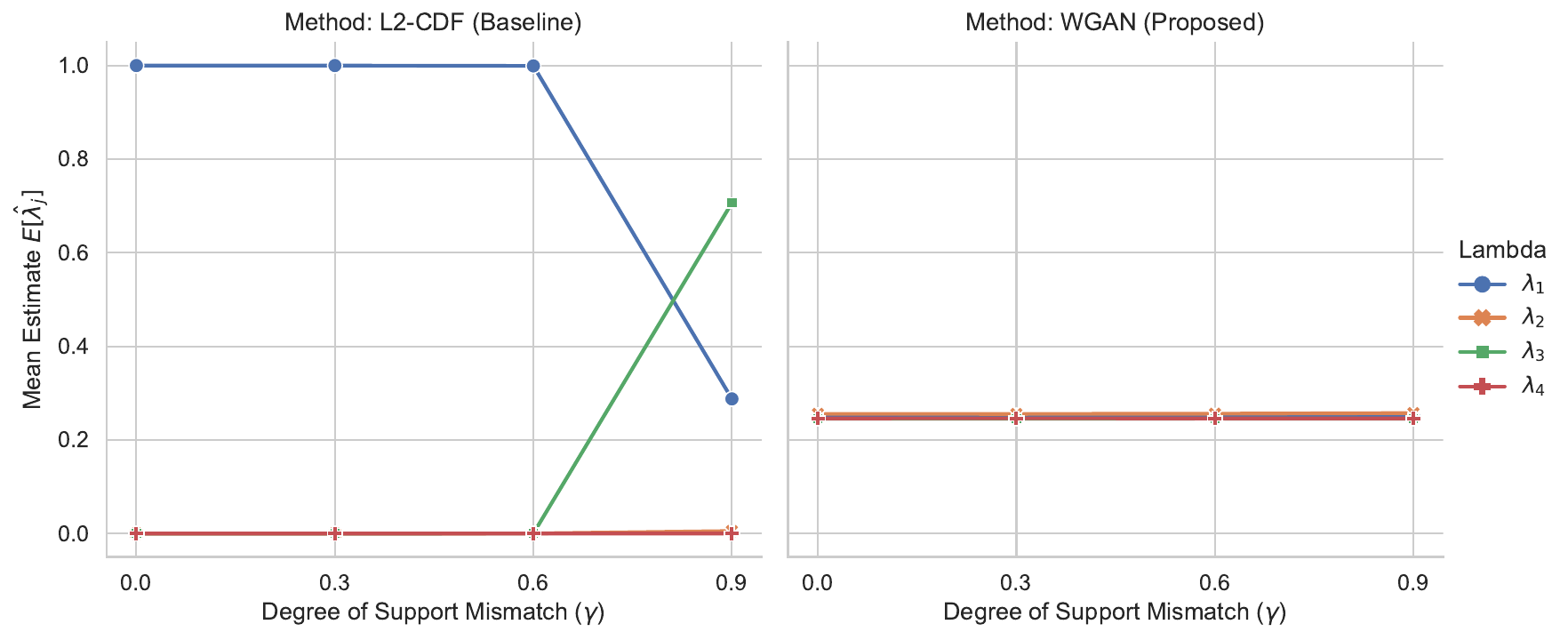}
    \caption{\textbf{Mean Weight Assignment under Support Mismatch.} The panels depict the average weight assigned to each donor ($J=4$). The $L_2$ solver (left) yields erratic estimates. The WGAN-E solver (right) converges to the uniform prior ($\lambda_j \approx 0.25$) as the transport cost becomes uniform across donors.}
    \label{fig:mean_behavior}
\end{figure}

Figure \ref{fig:mean_behavior} confirms that the stability is driven by the regularized transport geometry. Under disjoint supports, the Wasserstein distance strictly penalizes the geometric distance between probability masses. As $\gamma$ increases, the transport cost from any donor to the target becomes approximately equal (dominated by the gap $\gamma$). In this regime, the data term $W_1(\cdot)$ provides no discriminatory power, and the objective function is dominated by the entropy term. Consequently, the estimator converges to the barycenter of the simplex (uniform weights), properly signaling that the data is uninformative. The $L_2$ estimator, lacking both informative gradients and regularization, produces spurious sparsity ($\lambda_j \in \{0,1\}$), creating a false sense of identification.

\subsection{Structural Misspecification: Mixture Recovery}
\label{sec:sim_bimodal}

The third simulation investigates a fundamental theoretical limitation of standard distributional synthetic controls: model misspecification arising from multimodal outcomes. We construct a scenario where the target distribution is bimodal, representing a mixture of sub-populations, while the donor pool consists exclusively of unimodal distributions. This design exposes the structural inability of quantile-averaging methods to recover complex densities that do not satisfy location-scale preservation constraints.

\subsubsection{Data Generating Process and Quantile Aggregation Bias}

The target distribution $P_t^{Treated}$ is generated as an equiprobable mixture of two distinct Poisson measures :
\begin{equation}
    P_t^{Treated} = 0.5 \cdot \mathcal{P}(5) + 0.5 \cdot \mathcal{P}(20).
\end{equation}
This specification yields a discrete Probability Mass Function with distinct modes centered at 5 and 20. The donor pool consists of $J=4$ unimodal Poisson distributions $P_{j,t} = \mathcal{P}(\lambda_j)$, where the intensity parameters are given by $\bm{\lambda}_{donors} = \{2, 12, 18, 25\}$.

We benchmark the proposed WGAN-E solver against the standard $W_2$-Quantile estimator. The benchmark constructs the synthetic control by minimizing the $L_2$ distance between the quantile functions of the treated unit and a convex combination of donor quantiles :
\begin{equation}
    \hat{\bm{\lambda}}^{W2-Q} = \arg\min_{\bm{\lambda} \in \Delta^{J-1}} \left\| \sum_{j=1}^J \lambda_j Q_{P_j} - Q_{P_{Treated}} \right\|_2^2.
\end{equation}
This formulation relies on the implicit assumption that the quantile function of the mixture equals the mixture of the quantile functions. However, the mapping from a probability measure to its quantile function is non-linear with respect to convex combinations. Formally, for a weight vector $\bm{\lambda}$, the following inequality generally holds :
\begin{equation}
    Q\left(\sum_{j=1}^J \lambda_j P_j\right) \neq \sum_{j=1}^J \lambda_j Q(P_j).
\end{equation}
Consequently, the quantile averaging approach is theoretically misspecified for multimodal targets constructed from unimodal donors. In contrast, the WGAN-E estimator optimizes weights directly over the space of probability measures, treating the synthetic control as a statistical mixture $\sum \lambda_j P_j$ rather than a mixture of inverses.

\subsubsection{Results: Recovery of Bimodal Topology}

Figure \ref{fig:pmf_bimodal} compares the capacity of both estimators to reconstruct the target Probability Mass Function.

\begin{figure}[H]
    \centering
    \includegraphics[width=1\textwidth]{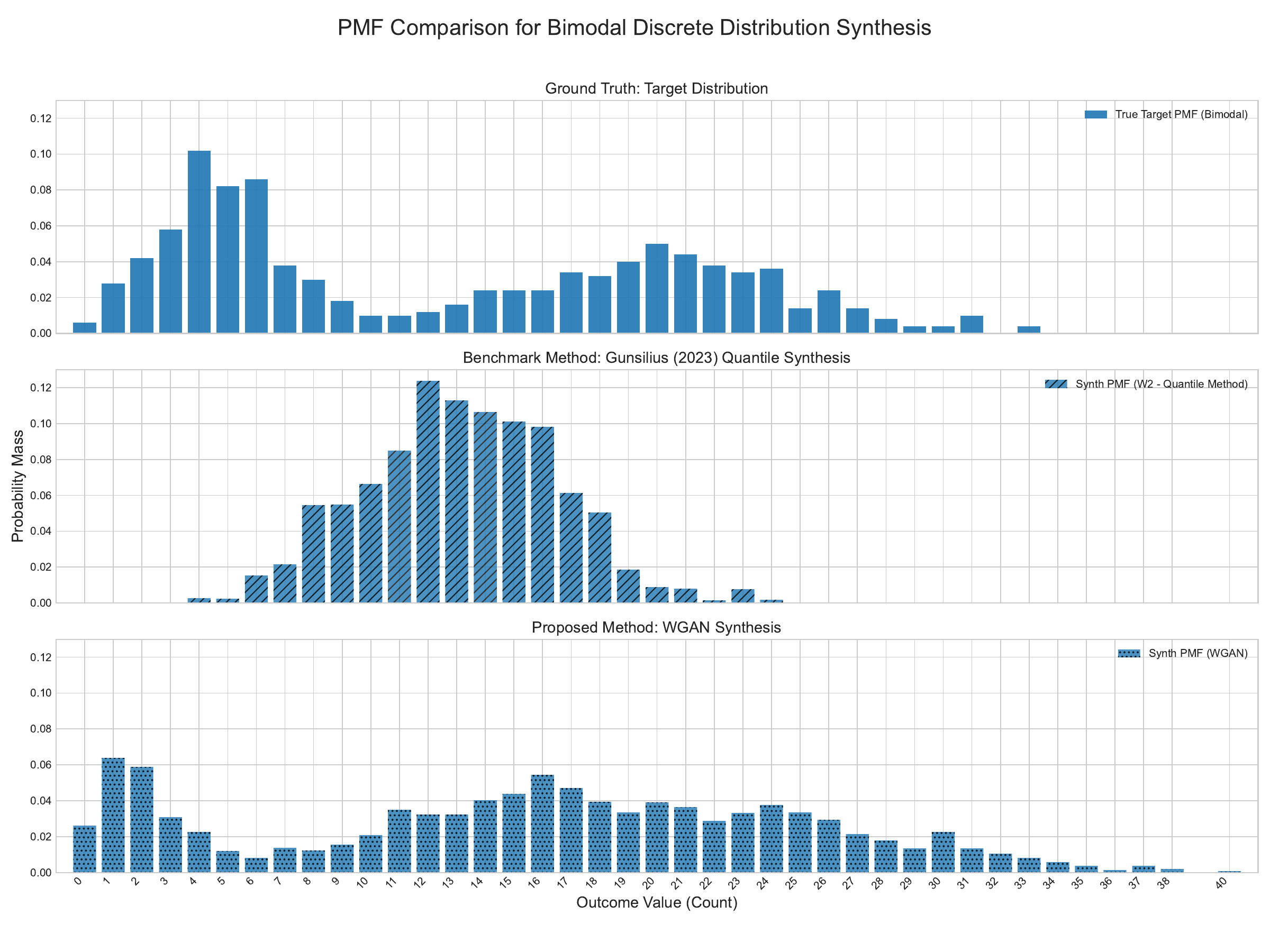}
    \caption{\textbf{Recovery of Bimodal Probability Mass Function.} The top panel displays the ground truth target distribution. The middle panel shows the synthetic distribution from the $W_2$-Quantile benchmark, which fails to capture the bimodality and produces a unimodal artifact. The bottom panel shows the WGAN-E synthesis, which successfully recovers the two distinct modes via a correct linear mixture of donors.}
    \label{fig:pmf_bimodal}
\end{figure}

The results demonstrate a structural divergence between the two methods. The $W_2$-Quantile solver displayed in the middle panel fails to replicate the bimodal structure. By averaging the quantiles of unimodal donors, the estimator produces a single, smoothed unimodal distribution centered between the two true modes. This outcome is a statistical artifact that represents neither sub-population, confirming that quantile-based synthesis requires the restrictive assumption that the counterfactual belongs to the same location-scale family as the donors .

Conversely, the proposed WGAN-E solver shown in the bottom panel successfully recovers the bimodal topology of the target. By optimizing the Wasserstein distance over the convex hull of donor measures, the estimator correctly identifies a weight vector that mixes donors with low intensity parameters around 5 and high intensity parameters around 20 to reconstruct the two distinct clusters. This result establishes that the measure-based approach remains consistent in scenarios involving complex or polarized distributions where quantile aggregation fails.

\subsection{Extension to Multivariate Domains}
\label{sec:sim_multivariate}

Finally, we demonstrate a decisive theoretical and computational advantage of the WGAN-based framework: its natural extension to multivariate outcome distributions. Standard distributional synthetic control methods relying on the inverse probability transform encounter fundamental definition challenges in dimensions $d \ge 2$. This limitation arises because Euclidean space lacks a canonical total ordering, rendering the definition of a multivariate quantile function theoretically ambiguous and computationally intractable . Consequently, existing quantile-based estimators cannot be straightforwardly applied to joint distributions of multiple outcomes, such as inflation and unemployment, or inequality and growth.

In contrast, the Wasserstein distance requires only a metric space and scales to high-dimensional distributions without structural modification. We validate this property through a bivariate proof-of-concept where the outcome lies in $\mathbb{R}^2$ .

\subsubsection{Multivariate Data Generating Process}

We design a scenario where the donor pool consists of $J=4$ bivariate Gaussian distributions $P_j = \mathcal{N}(\bm{\mu}_j, \bm{\Sigma})$. The mean vectors $\bm{\mu}_j$ are positioned symmetrically at the four corners of a square grid, specifically $\{-3, 3\} \times \{-3, 3\}$. All units share a covariance matrix $\bm{\Sigma}$ with a correlation coefficient $\rho=0.5$. The target distribution $P_t^{Treated}$ is constructed as a known mixture of these four donor distributions using the ground-truth weight vector $\bm{\lambda}^{true} = (0.15, 0.25, 0.35, 0.25)^\top$ .

To solve this problem, we adapt the WGAN-E architecture by modifying the critic network to accept two-dimensional inputs $D: \mathbb{R}^2 \to \mathbb{R}$. The generator minimizes the $W_1$ distance as defined in Equation \ref{w1}. We note that the $W_2$-Quantile benchmark is excluded from this simulation as the estimator is undefined for multivariate data .

\subsubsection{Visual Reconstruction of Joint Distributions}

\begin{figure}[htbp]
    \centering
    \includegraphics[width=1\linewidth]{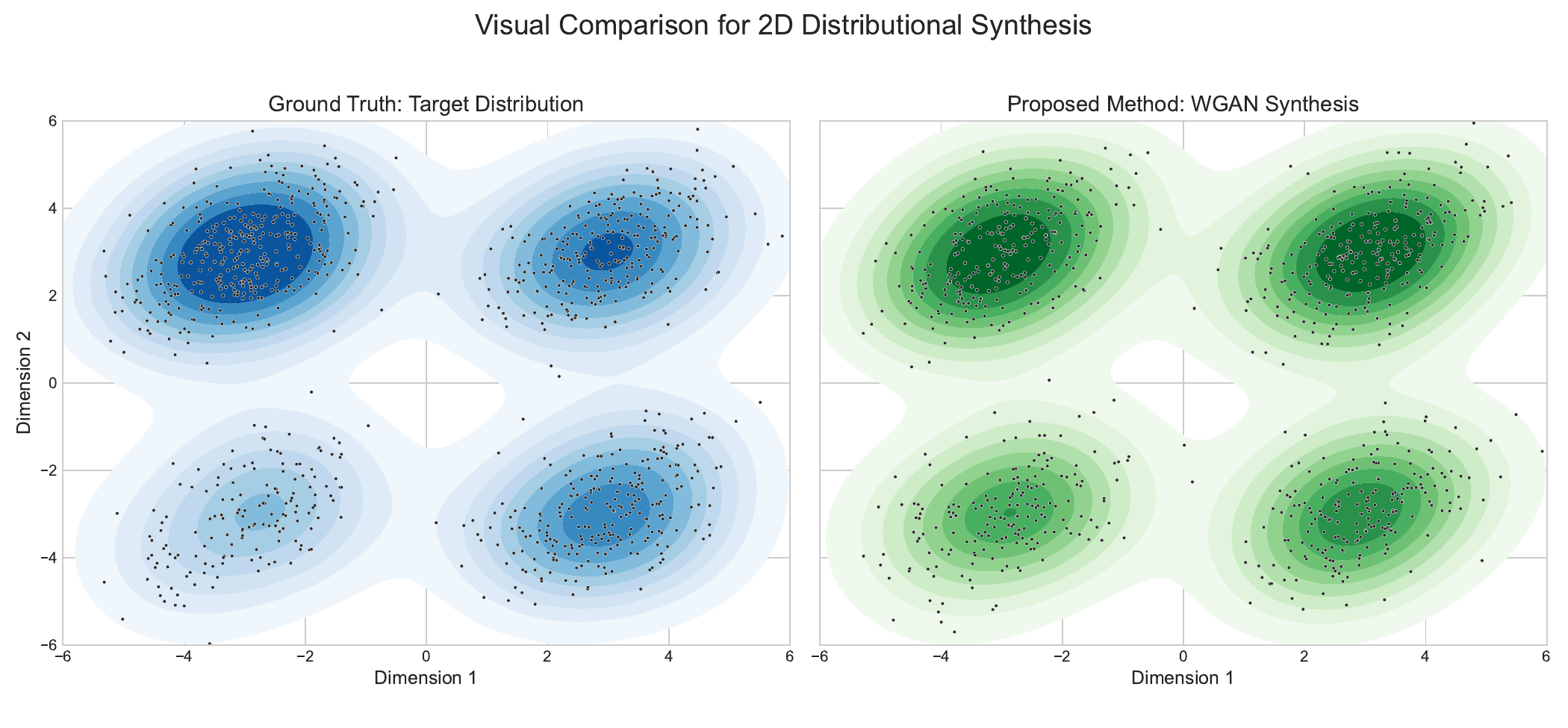}
    \caption{\textbf{Visual Comparison for 2D Distributional Synthesis.} The left panel displays the kernel density estimate of the ground truth target distribution, which exhibits four distinct modes corresponding to the Gaussian components. The right panel displays the synthetic distribution generated by the WGAN-E solver. The estimator successfully identifies the multimodal structure and accurately replicates both the location and the relative probability mass of each cluster.}
    \label{fig:2d_synthesis}
\end{figure}

Figure \ref{fig:2d_synthesis} provides a qualitative validation of the synthesis. The side-by-side comparison reveals a compelling match. The Ground Truth panel displays the multimodal structure of the target characterized by four distinct clusters with varying probability masses. The WGAN-E synthesis accurately replicates this topology. The solver correctly identifies the location of all four modes and assigns the correct relative density to each cluster, reflecting the underlying weights $\bm{\lambda}^{true}$.

This result serves as a proof-of-concept that the adversarial framework effectively circumvents the dimensionality curse of quantile-based methods. By leveraging neural networks to approximate the Kantorovich potential, our method decouples the computational complexity from the geometric complexity of the outcome space. This enables robust distributional synthesis in multivariate settings where traditional sorting-based methods are inapplicable .

\section{Empirical Application: A Stress Test on Structural Breaks}
\label{sec:application}

To demonstrate the finite-sample robustness of the proposed WGAN-DSC estimator, we apply it to a canonical stress test scenario: the Kansas tax experiment of 2012. While originating from fiscal policy, this episode serves as an ideal testing ground for financial econometrics because it exhibits the stylized facts often encountered in asset pricing and risk management: a massive, exogenous structural shock, heavy-tailed distributional dynamics, and a potential regime switch where the post-treatment support deviates from the historical convex hull.

The 2012 reform represents a distinct regime switch intended to alter the incentive structure for capital accumulation. From an econometric perspective, this intervention tests whether the estimator can distinguish between a shift in the first moment and a fundamental deformation of the distribution's geometry. While theoretical proponents anticipated a first-order stochastic dominance shift, the empirical question is whether the shock primarily induced idiosyncratic volatility and tail risk without a corresponding mean shift. Standard difference-in-differences estimators, which focus on conditional means, fail to capture the geometry of such shifts for the same reason they limit the analysis of asymmetric financial returns.

We construct empirical measures of log real household income using micro-level data from the IPUMS CPS March Supplement (2005--2016). Crucially, this application tests the estimator's ability to handle support mismatch. Similar to liquidity gaps in high-frequency trading, the post-treatment distribution may drift into regions sparsely covered by the pre-treatment donor pool. As established in Section \ref{economicframe}, while Euclidean metrics would face vanishing gradients in such disjoint regimes, the Wasserstein metric leverages the underlying metric space to provide a meaningful signal of distance regardless of overlap.

\begin{figure}[htbp]
    \centering
    \includegraphics[width=1\linewidth]{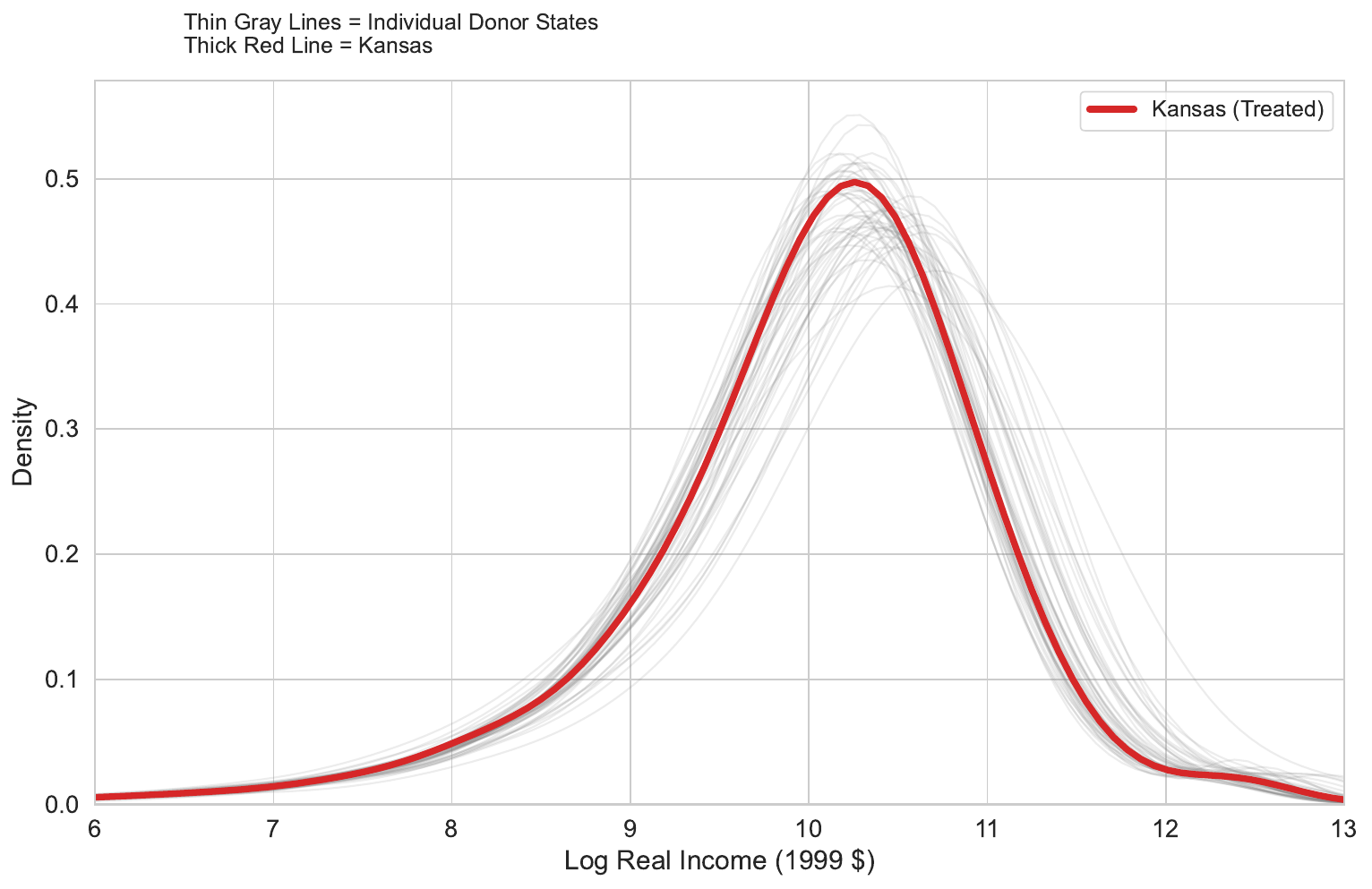}
    \caption{\textbf{Common Support Verification.} The plot displays the kernel density estimates of log real income for the treated unit (Kansas, thick red line) and all individual donor states (thin gray lines) averaged over the pre-treatment period. The visualization confirms that the support of the target distribution lies well within the convex hull of the donor pool, satisfying the geometric identification condition.}
    \label{fig:kansas_pre}
\end{figure}

Figure \ref{fig:kansas_pre} verifies the geometric identification condition. The income distribution of Kansas is enveloped by the donor pool, indicating that a convex combination of donors can theoretically reconstruct the target distribution. Unlike standard methods that match only mean moments, our approach requires the donor pool to span the support of the treated unit's probability measure to ensure valid transport.

We apply the WGAN-E estimator to construct a synthetic counterfactual by minimizing the regularized Wasserstein-1 distance between the empirical measure of Kansas and a weighted barycenter of donor states for the pre-treatment years 2005--2011.

\begin{figure}[H]
    \centering
    \includegraphics[width=1\linewidth]{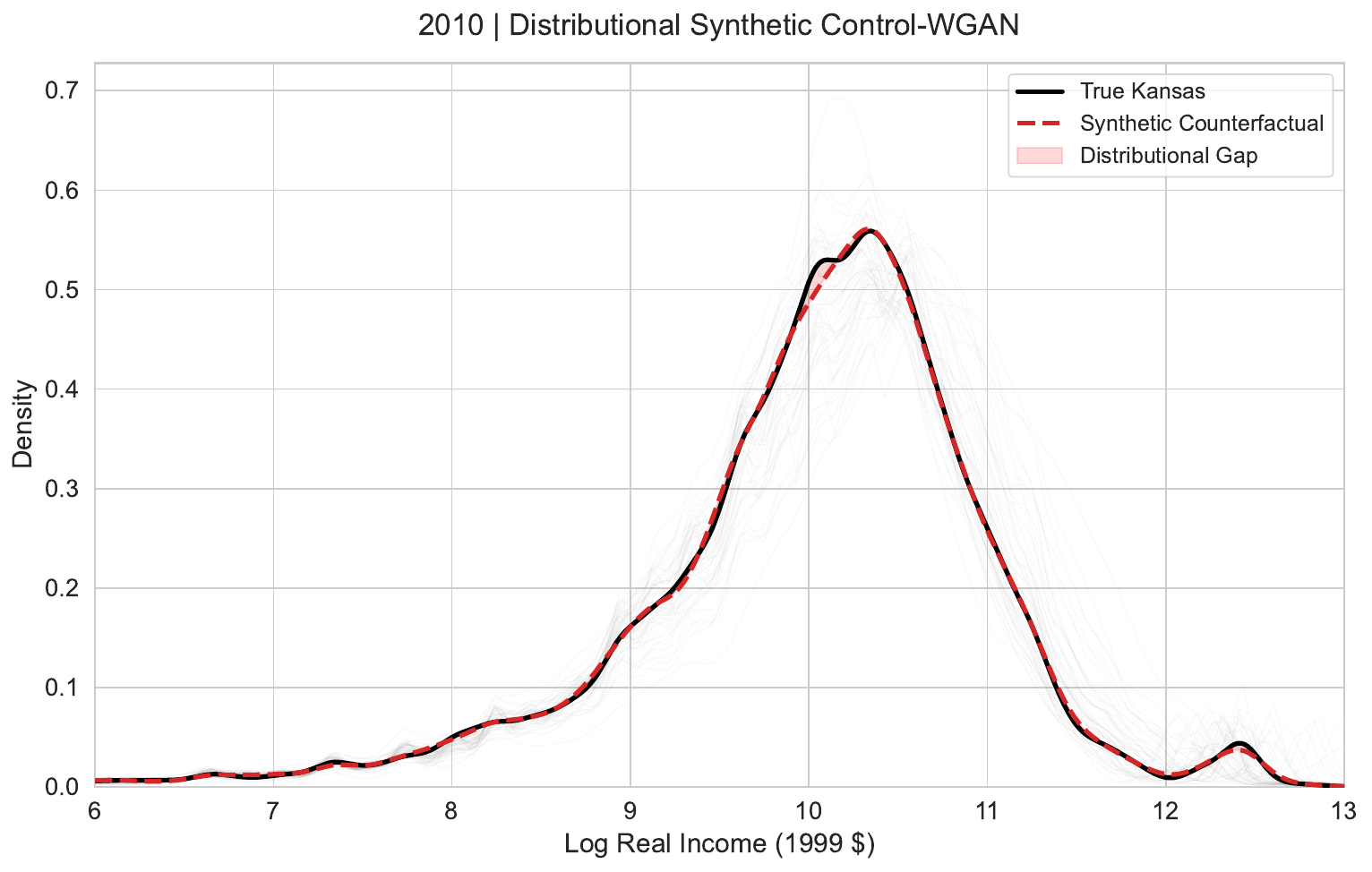}
    \caption{\textbf{Pre-treatment Goodness-of-Fit (2010).} The figure compares the observed distribution of Kansas (solid black line) against the WGAN-generated synthetic counterfactual (dashed red line) for a placebo year, 2010. The tight overlap demonstrates the estimator's ability to recover the latent factor structure driving the income distribution prior to the intervention.}
    \label{fig:kansas_fit}
\end{figure}

Figure \ref{fig:kansas_fit} illustrates the goodness-of-fit for 2010, ensuring the validity of the learned weights. The WGAN-based synthetic control reproduces the observed distribution with remarkable precision, capturing not just the location but also the higher-order moments (skewness and kurtosis). This validates that the estimated donor weights $\bm{\lambda}^*$ successfully recover the latent factor structure driving the regional economy prior to the shock.

We estimate the counterfactual distributions for the post-treatment years 2013--2015 to quantify the distributional impact. The results summarized in Figure \ref{fig:kansas_results} reveal a distinct structural decoupling that mean-variance models would miss.

\begin{figure}[htbp]
    \centering
    \begin{subfigure}[b]{0.48\textwidth}
        \centering
        \includegraphics[width=\textwidth]{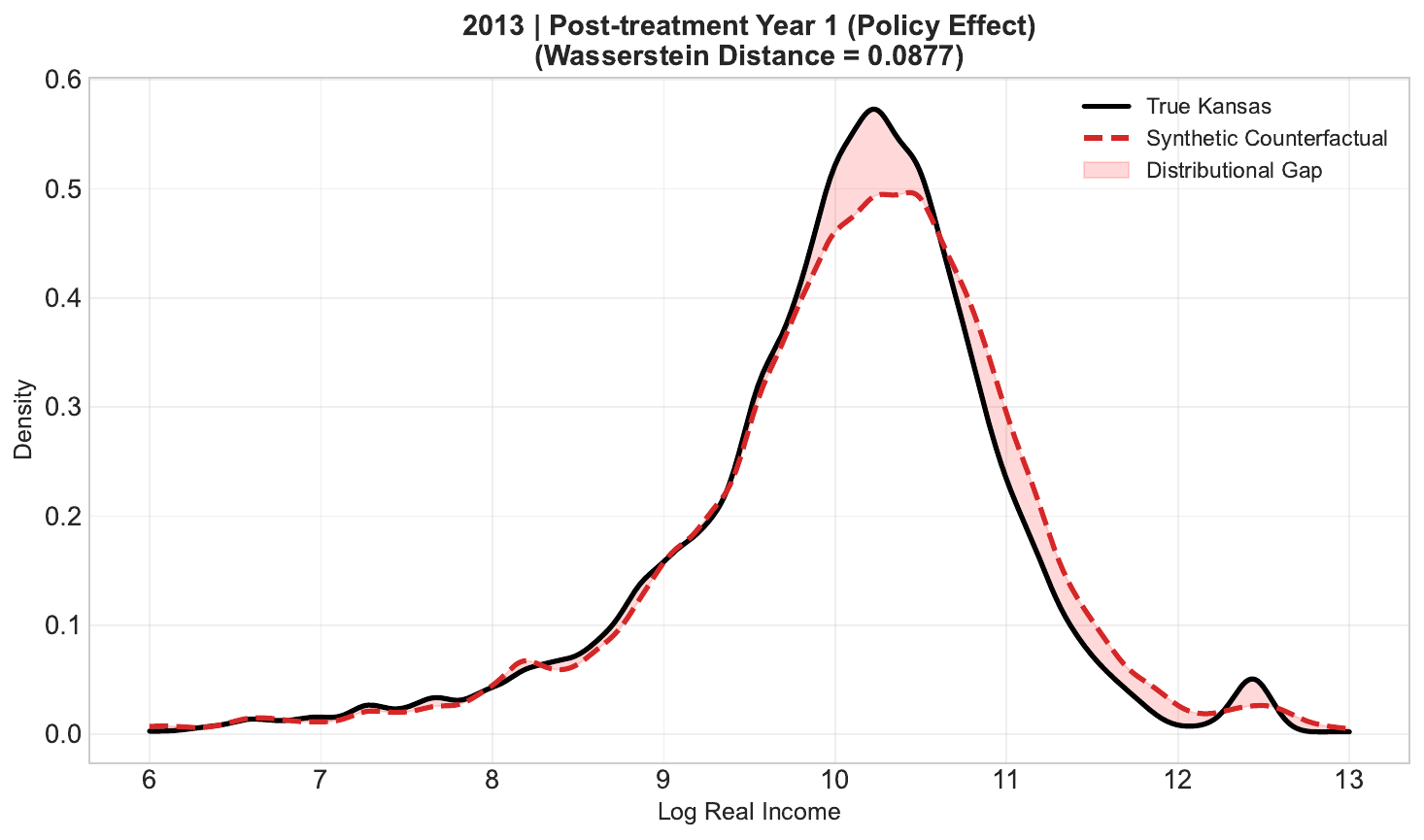}
        \caption{2013 ($W_1 = 0.0877$)}
    \end{subfigure}
    \hfill
    \begin{subfigure}[b]{0.48\textwidth}
        \centering
        \includegraphics[width=\textwidth]{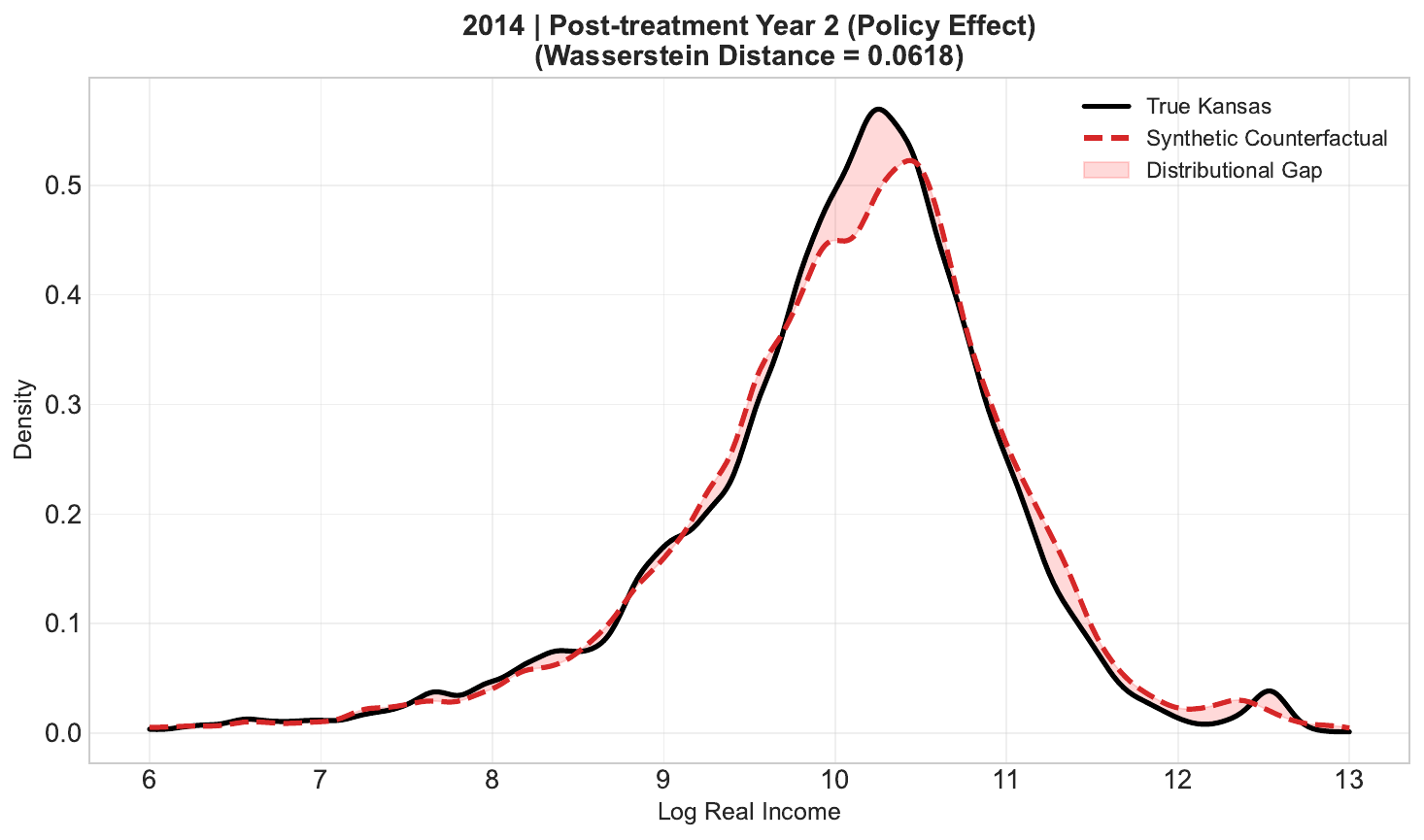}
        \caption{2014 ($W_1 = 0.0618$)}
    \end{subfigure}
    \vspace{1em}
    \begin{subfigure}[b]{0.48\textwidth}
        \centering
        \includegraphics[width=\textwidth]{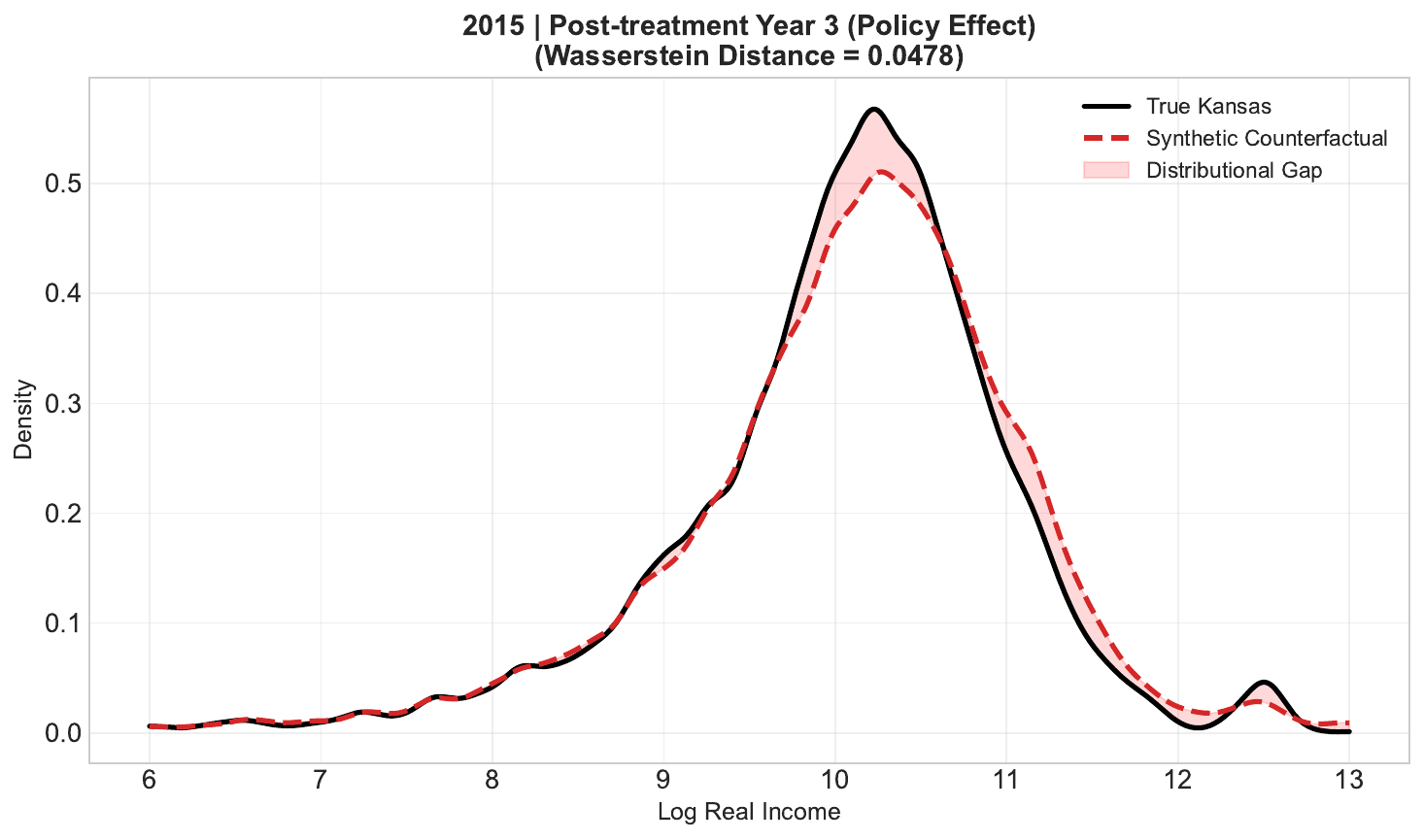}
        \caption{2015 ($W_1 = 0.0478$)}
    \end{subfigure}
    \caption{\textbf{Post-Treatment Distributional Effects.} The panels display the observed outcome (solid black) versus the synthetic counterfactual (dashed red) for the years following the 2012 regime switch. The shaded regions represent the transport mass required to align the distributions, visually quantifying the distributional gap. The substantial Wasserstein distances indicate a structural deviation from the counterfactual trend.}
    \label{fig:kansas_results}
\end{figure}

In 2013, the first full year of the new regime, we observe an immediate distributional shock with the Wasserstein distance spiking to $0.088$. Visually, the synthetic counterfactual exhibits a probability mass concentration distinct from the realized outcome. By 2014 and 2015, while the magnitude of the Wasserstein distance slightly attenuates, the structural gap persists.

The shaded distributional gap highlights that the policy did not simply shift the distribution uniformly; rather, the actual distribution lags behind the counterfactual in specific quantiles. This pattern suggests a heterogeneous response analogous to an idiosyncratic volatility shock in asset markets. The successful recovery of the counterfactual distribution in this setting suggests the WGAN-DSC is well-suited for financial applications requiring tail risk attribution and illiquid asset valuation, where capturing the morphological shift of the stochastic process is paramount.

\section{Conclusion}
\label{sec:conclusion}

This paper proposes a robust framework for Distributional Synthetic Controls (DSC) grounded in the geometry of Optimal Transport (OT). While the extension of the synthetic control method from mean outcomes to entire probability distributions represents a significant leap in causal inference, existing approaches relying on $L_2$-quantile minimization remain vulnerable to the geometric limitations of Euclidean metrics. We demonstrate that such methods implicitly assume overlapping supports and unimodal structures, rendering them fragile in the presence of outliers, support mismatch, or complex multimodal outcomes.

Our methodological contribution is to reformulate the synthetic control problem as the minimization of the Wasserstein-1 distance between probability measures. By leveraging the Kantorovich-Rubinstein duality within a Generative Adversarial Network (WGAN) architecture, we provide a feasible computational strategy that scales to continuous and high-dimensional outcomes.

The theoretical and simulation results presented herein establish three decisive advantages of this approach. First, the Wasserstein metric provides valid gradient signals even when the supports of the treated and donor distributions are disjoint, resolving the support mismatch failure that plagues standard density-matching estimators. Second, by penalizing transport cost rather than vertical discrepancies, our estimator exhibits superior robustness to heavy-tailed contamination. Third, identifying the synthetic control as a linear mixture of measures—rather than a mixture of quantile functions—allows for the accurate recovery of complex, multimodal counterfactuals where traditional quantile averaging fails.

Practically, this framework expands the scope of the synthetic control method to a broader class of empirical questions. Researchers can now rigorously evaluate the impact of interventions on inequality, polarization, and multidimensional welfare metrics without being constrained by the rigid distributional assumptions of prior methods. As causal inference increasingly moves towards characterizing heterogeneity beyond the mean, the integration of OT principles offers a principled and robust path forward.

\newpage

\addappheadtotoc
\bibliography{main}
\newpage

\newpage
\section*{Appendix A: Mathematical Proofs}

\subsection*{A.1 Proof of Theorem 1 (Counterfactual Recovery)}

\begin{proof}
We aim to show that for any post-intervention period $t > T_0$, the distance between the synthetic control distribution and the true counterfactual distribution is zero.

Let $P_{1t}^N$ denote the counterfactual distribution of the treated unit, and let $P_{\lambda^*, t} = \sum_{j=2}^{J+1} \lambda_j^* P_{jt}$ denote the synthetic control distribution constructed using the optimal weights $\lambda^*$.

Using the Latent Factor Generation process defined in Assumption 1, we can express the observed distributions as the push-forward of the latent factors via the transport map $\mathcal{T}_t$. Specifically, $P_{jt} = \mathcal{T}_t(\mu_j)$ for all units $j$.

The Wasserstein-1 distance between the synthetic and counterfactual distributions is:
\begin{equation}
W_1\left( P_{\lambda^*, t}, P_{1t}^N \right) = W_1\left( \sum_{j=2}^{J+1} \lambda_j^* \mathcal{T}_t(\mu_j), \mathcal{T}_t(\mu_1) \right)
\end{equation}

\paragraph{Step 1: Affine Linearity of Mixtures.}
The mixture operation commutes with the transport map under the structural model. Since the synthetic control is defined as a linear mixture of measures, and the latent structure implies that the observed mixture arises from the mixture of latent factors transformed by the aggregate shock, we have:
\begin{equation}
\sum_{j=2}^{J+1} \lambda_j^* \mathcal{T}_t(\mu_j) = \mathcal{T}_t\left( \sum_{j=2}^{J+1} \lambda_j^* \mu_j \right)
\end{equation}
Substituting this into the distance equation yields:
\begin{equation}
W_1\left( P_{\lambda^*, t}, P_{1t}^N \right) = W_1\left( \mathcal{T}_t\left( \sum_{j=2}^{J+1} \lambda_j^* \mu_j \right), \mathcal{T}_t(\mu_1) \right)
\end{equation}

\paragraph{Step 2: Structural Stability (Scaled Isometry).}
We invoke Assumption 2, which states that the transport map $\mathcal{T}_t$ satisfies the scaled isometry condition $W_1(\mathcal{T}_t(\nu_1), \mathcal{T}_t(\nu_2)) = \kappa_t W_1(\nu_1, \nu_2)$. Applying this property:
\begin{equation}
W_1\left( \mathcal{T}_t\left( \sum_{j=2}^{J+1} \lambda_j^* \mu_j \right), \mathcal{T}_t(\mu_1) \right) = \kappa_t W_1\left( \sum_{j=2}^{J+1} \lambda_j^* \mu_j, \mu_1 \right)
\end{equation}

\paragraph{Step 3: Pre-treatment Matching.}
By the condition of Theorem 1, the optimal weights $\lambda^*$ ensure that the treated unit's latent factors lie within the convex hull of the donors, specifically $\mu_1 = \sum_{j=2}^{J+1} \lambda_j^* \mu_j$. Consequently, the distance in the latent space is strictly zero:
\begin{equation}
W_1\left( \sum_{j=2}^{J+1} \lambda_j^* \mu_j, \mu_1 \right) = 0
\end{equation}

Combining Steps 2 and 3, we obtain:
\begin{equation}
W_1\left( P_{\lambda^*, t}, P_{1t}^N \right) = \kappa_t \cdot 0 = 0
\end{equation}
Thus, the synthetic control perfectly recovers the counterfactual distribution $P_{1t}^N$ for all $t > T_0$.
\end{proof}

\section*{Appendix B: Gradient Derivations}

\subsection*{B.1 Proof of Proposition \ref{prop:gradients}}

\begin{proof}
\textbf{Part 1 ($L_2$ Case):}
The squared $L_2$ distance between densities is given by the functional $\mathcal{L}_{L2}(\lambda) = \int (p_1(x) - \sum_{j} \lambda_j p_j(x))^2 dx$. Expanding the square term yields:
\begin{equation}
\mathcal{L}_{L2}(\lambda) = \int p_1(x)^2 dx - 2 \sum_{j} \lambda_j \int p_1(x)p_j(x) dx + \int \left(\sum_{j} \lambda_j p_j(x)\right)^2 dx
\end{equation}
Under the assumption of disjoint supports ($\text{supp}(P_{1}) \cap \text{supp}(P_{j}) = \emptyset$ for all $j$), the cross-terms $\int p_1(x)p_j(x) dx$ are strictly zero. Furthermore, the remaining terms only involve integrals over the donor supports or the target support independently.
Crucially, variations in $\lambda$ only reallocate mass among donors within the donor support region. Since this region is disjoint from the target, local perturbations in weights do not reduce the distance to the target. Mathematically, the derivative $\frac{\partial \mathcal{L}_{L2}}{\partial \lambda_j}$ depends only on the interactions between donors, providing no signal regarding the direction or distance towards the target $P_1$. Thus, the gradient is uninformative for reducing the gap between the disjoint supports.

\textbf{Part 2 ($W_1$ Case):}
The Wasserstein-1 distance is defined by the Kantorovich-Rubinstein dual:
\begin{equation}
W_1(P_{\lambda}, P_1) = \sup_{f \in Lip_1(\mathcal{X})} \Phi(\lambda, f)
\end{equation}
where $\Phi(\lambda, f) = \mathbb{E}_{x \sim P_1}[f(x)] - \sum_{j} \lambda_j \mathbb{E}_{x \sim P_j}[f(x)]$.
The function $\Phi(\lambda, f)$ is linear in $\lambda$ and concave in $f$. Since the set of 1-Lipschitz functions $Lip_1(\mathcal{X})$ is compact (under appropriate topology) and $\Phi$ is continuous, we can apply Danskin's Theorem (or the Envelope Theorem). The sub-gradient of the supremum function value with respect to $\lambda_j$ is equal to the partial derivative of $\Phi$ evaluated at the optimal dual potential $f^*$:
\begin{equation}
\frac{\partial W_1}{\partial \lambda_j} = \frac{\partial \Phi(\lambda, f^*)}{\partial \lambda_j} = -\mathbb{E}_{x \sim P_j}[f^*(x)]
\end{equation}
Since $f^* \in Lip_1(\mathcal{X})$ is a 1-Lipschitz function that approximates the underlying metric distance (roughly, $|f^*(x) - f^*(y)| \approx d(x, y)$), the term $\mathbb{E}_{P_j}[f^*]$ encodes the geometric "transport cost" or position of donor $j$ relative to the target. This value is distinct for each donor depending on its distance to the target. Therefore, the gradient is non-zero and points in the direction of the donors that are spatially closest to the target, regardless of whether the supports overlap.
\end{proof}

\section*{Appendix C: Additional Mathematical Proofs}

\subsection*{C.1 Proof of Theorem \ref{thm:identification} (Global Identification)}

\begin{proof}
    Let the objective function be $Q_\eta(\lambda) = g(\lambda) + \eta h(\lambda)$, defined on the simplex $\Delta^{J-1}$, where $g(\lambda) = W_1(P_\lambda, P_1)$ and $h(\lambda) = \sum_{j=2}^{J+1} \lambda_j \log \lambda_j$.

    \textbf{1. Convexity of the Transport Cost $g(\lambda)$.}
    By the Kantorovich-Rubinstein duality theorem, we express $g(\lambda)$ as:
    \begin{equation}
        g(\lambda) = \sup_{f \in Lip_1(\mathcal{X})} \left( \mathbb{E}_{P_1}[f] - \sum_{j=2}^{J+1} \lambda_j \mathbb{E}_{P_j}[f] \right).
    \end{equation}
    Define $\mathcal{L}(\lambda, f) := \mathbb{E}_{P_1}[f] - \sum_{j=2}^{J+1} \lambda_j \mathbb{E}_{P_j}[f]$. For any fixed $f \in Lip_1(\mathcal{X})$, the map $\lambda \mapsto \mathcal{L}(\lambda, f)$ is affine. Since the pointwise supremum of a family of affine functions is convex, it follows that:
    \begin{equation}
        g(\theta \lambda + (1-\theta)\lambda') \le \theta g(\lambda) + (1-\theta)g(\lambda'), \quad \forall \lambda, \lambda' \in \Delta^{J-1}, \forall \theta \in [0,1].
    \end{equation}
    Thus, $g(\lambda)$ is convex on $\Delta^{J-1}$.

    \textbf{2. Strict Convexity of the Regularizer $h(\lambda)$.}
    The Hessian of $h(\lambda)$ with respect to $\lambda$ is given by:
    \begin{equation}
        \nabla^2 h(\lambda) = \text{diag}\left(\frac{1}{\lambda_2}, \dots, \frac{1}{\lambda_{J+1}}\right).
    \end{equation}
    For any $\lambda$ in the relative interior of $\Delta^{J-1}$ (where $\lambda_j > 0$ for all $j$), the quadratic form satisfies:
    \begin{equation}
        v^\top \nabla^2 h(\lambda) v = \sum_{j=2}^{J+1} \frac{v_j^2}{\lambda_j} > 0, \quad \forall v \in \mathbb{R}^J \setminus \{0\}.
    \end{equation}
    Therefore, $\nabla^2 h(\lambda) \succ 0$, and $h(\lambda)$ is strictly convex.

    \textbf{3. Global Uniqueness.}
    Since $\eta > 0$, the total objective $Q_\eta(\lambda)$ is the sum of a convex function $g$ and a strictly convex function $\eta h$. Consequently, $Q_\eta(\lambda)$ is strictly convex on $\Delta^{J-1}$. Since $\Delta^{J-1}$ is a compact convex set, the minimizer
    \begin{equation}
        \lambda^* = \arg\min_{\lambda \in \Delta^{J-1}} Q_\eta(\lambda)
    \end{equation}
    exists and is unique. Assumption \ref{ass:affine} ensures the injectivity of the mixture map $\Phi$, guaranteeing that this unique weight vector corresponds to a structurally unique synthetic distribution.
\end{proof}

\subsection*{C.2 Proof of Theorem \ref{thm:consistency}}

\begin{proof}
    We verify the uniform convergence condition $\sup_{\lambda \in \Delta^{J-1}} |\hat{Q}_N(\lambda) - Q_\eta(\lambda)| \xrightarrow{p} 0$. Other conditions for M-estimator consistency (compactness, identification, continuity) hold by construction and Theorem \ref{thm:identification}.

    \textbf{Step 1: Decomposition via Triangle Inequalities.}
    The regularization term $\eta \Omega(\lambda)$ is deterministic and cancels out. Let $P_\lambda = \sum \lambda_j P_j$ and $\hat{P}_{\lambda,N} = \sum \lambda_j \hat{P}_{j,N}$. By the triangle inequality:
    \begin{align}
        |\hat{Q}_N(\lambda) - Q_\eta(\lambda)| &= |W_1(\hat{P}_{\lambda,N}, \hat{P}_{1,N}) - W_1(P_\lambda, P_1)| \\
        &\le W_1(\hat{P}_{\lambda,N}, P_\lambda) + W_1(\hat{P}_{1,N}, P_1).
    \end{align}

    \textbf{Step 2: Convexity Bound.}
    Using the joint convexity of the Wasserstein metric $W_1(\cdot, \cdot)$:
    \begin{equation}
        W_1(\hat{P}_{\lambda,N}, P_\lambda) = W_1\left(\sum_{j=2}^{J+1} \lambda_j \hat{P}_{j,N}, \sum_{j=2}^{J+1} \lambda_j P_j\right) \le \sum_{j=2}^{J+1} \lambda_j W_1(\hat{P}_{j,N}, P_j).
    \end{equation}
    Substituting this into the inequality from Step 1:
    \begin{equation}
        |\hat{Q}_N(\lambda) - Q_\eta(\lambda)| \le \sum_{j=2}^{J+1} \lambda_j W_1(\hat{P}_{j,N}, P_j) + W_1(\hat{P}_{1,N}, P_1).
    \end{equation}

    \textbf{Step 3: Uniformity Argument.}
    Since $\lambda \in \Delta^{J-1}$, we have $0 \le \lambda_j \le 1$ for all $j$. Taking the supremum over the simplex eliminates the dependence on $\lambda$:
    \begin{equation}
        \sup_{\lambda \in \Delta^{J-1}} |\hat{Q}_N(\lambda) - Q_\eta(\lambda)| \le \sum_{j=2}^{J+1} W_1(\hat{P}_{j,N}, P_j) + W_1(\hat{P}_{1,N}, P_1).
    \end{equation}

    \textbf{Step 4: Convergence of Empirical Measures.}
    Under Assumption \ref{ass:moments} (Finite Moments), Varadarajan's Theorem guarantees the convergence of empirical measures in the Wasserstein distance:
    \begin{equation}
        W_1(\hat{P}_{j,N}, P_j) \xrightarrow{a.s.} 0, \quad \forall j \in \{1, \dots, J+1\}.
    \end{equation}
    Consequently, the RHS of the inequality in Step 3 converges to 0 almost surely. Thus, $\sup_{\lambda} |\hat{Q}_N(\lambda) - Q_\eta(\lambda)| \xrightarrow{p} 0$, implying $\hat{\lambda}_N \xrightarrow{p} \lambda^*$.
\end{proof}

\subsection*{C.3 Proof of Theorem \ref{thm:normality} (Asymptotic Normality)}

\begin{proof}
    Let $\Psi_N(\lambda) = \nabla_\lambda \hat{Q}_N(\lambda)$ be the gradient of the empirical objective. Since $\hat{\lambda}_N$ is an interior solution (due to entropic regularization forcing weights away from boundaries), it satisfies the first-order condition $\Psi_N(\hat{\lambda}_N) = 0$.
    
    \textbf{Step 1: Taylor Expansion.}
    Expanding around $\lambda^*$:
    \begin{equation}
        0 = \Psi_N(\hat{\lambda}_N) = \Psi_N(\lambda^*) + \nabla_\lambda \Psi_N(\tilde{\lambda}) (\hat{\lambda}_N - \lambda^*),
    \end{equation}
    where $\tilde{\lambda}$ lies between $\hat{\lambda}_N$ and $\lambda^*$. By the consistency theorem and smoothness, $\nabla_\lambda \Psi_N(\tilde{\lambda}) \xrightarrow{p} H^*$, where $H^*$ is the positive definite Hessian (Assumption \ref{ass:hessian}). Thus:
    \begin{equation}
        \hat{\lambda}_N - \lambda^* = -(H^*)^{-1} \Psi_N(\lambda^*) + o_p(N^{-1/2}).
    \end{equation}

    \textbf{Step 2: Decomposition of the Score.}
    The empirical gradient at the truth, $\Psi_N(\lambda^*)$, decomposes into a stochastic term and an approximation bias:
    \begin{equation}
        \Psi_N(\lambda^*) = \underbrace{\frac{1}{N} \sum_{k=1}^N S(Y^{(k)}, \lambda^*)}_{\text{Empirical Process}} + \underbrace{\text{Bias}_N}_{\text{Sieve Error}}.
    \end{equation}
    Here, the score function $S(Y, \lambda^*) = \eta(1 + \log \lambda^*) + \nabla_\lambda W_1(P_{\lambda^*}, P_1)$ depends on the optimal dual potential $f^*$.

    \textbf{Step 3: Handling Bias and Variance.}
    Under Assumption \ref{ass:sieve_rate}, the bias term satisfies $\sqrt{N} \cdot \text{Bias}_N = O(\sqrt{N}\delta_N) = o(1)$.
    Under Assumption \ref{ass:donsker}, the empirical process term obeys the Central Limit Theorem:
    \begin{equation}
        \sqrt{N} \left( \frac{1}{N} \sum_{k=1}^N S(Y^{(k)}, \lambda^*) \right) \xrightarrow{d} \mathcal{N}(0, V),
    \end{equation}
    where $V = \mathbb{E}[S(Y, \lambda^*) S(Y, \lambda^*)^\top]$.

    \textbf{Conclusion.}
    Multiplying the Taylor expansion by $\sqrt{N}$ and substituting the components:
    \begin{equation}
        \sqrt{N}(\hat{\lambda}_N - \lambda^*) = -(H^*)^{-1} \cdot \mathcal{N}(0, V) + o_p(1) \xrightarrow{d} \mathcal{N}(0, (H^*)^{-1} V (H^*)^{-1}).
    \end{equation}
\end{proof}

\subsection*{C.4 Proof of Theorem \ref{thm:permutation} (Validity of Permutation Inference)}

\begin{proof}
    Let $\mathbf{Y}$ denote the full dataset of observed outcomes for all units over all time periods. Let $\mathcal{U} = \{1, \dots, J+1\}$ be the set of unit indices. 
    
    \textbf{1. Exchangeability under the Null.}
    The null hypothesis $H_0$ posits that the intervention has no effect on the outcome distribution. Under $H_0$, and assuming random assignment of the treatment, the joint distribution of the data is invariant to permutations of the unit labels. Specifically, the index of the treated unit, denoted by the random variable $I$, is uniformly distributed over $\mathcal{U}$:
    \begin{equation}
        \mathbb{P}(I = j) = \frac{1}{J+1}, \quad \forall j \in \mathcal{U}.
    \end{equation}

    \textbf{2. Distribution of the Test Statistic.}
    Let $\{\hat{\tau}_j\}_{j=1}^{J+1}$ be the collection of placebo statistics computed from the data $\mathbf{Y}$. Since the assignment $I$ is uniform on $\mathcal{U}$, the test statistic corresponding to the treated unit, $\hat{\tau}_I$, is equally likely to take any value from the set $\{\hat{\tau}_1, \dots, \hat{\tau}_{J+1}\}$.
    
    \textbf{3. Exact Size Control.}
    The p-value is defined as $P(\mathbf{Y}) = \frac{1}{J+1} \sum_{j=1}^{J+1} \mathbb{I}(\hat{\tau}_j \ge \hat{\tau}_I)$. 
    Let $R$ be the rank of $\hat{\tau}_I$ among $\{\hat{\tau}_j\}$ in descending order (where rank 1 is the largest). The p-value can be written as $P(\mathbf{Y}) = \frac{R}{J+1}$.
    
    Since $I$ is uniform, $R$ follows a discrete uniform distribution on $\{1, \dots, J+1\}$. For any significance level $\alpha \in [0, 1]$, let $k = \lfloor \alpha(J+1) \rfloor$. The event $\{p\text{-value} \le \alpha\}$ is equivalent to $\{R \le \alpha(J+1)\}$, or $\{R \le k\}$.
    
    The probability of this event is:
    \begin{equation}
        \mathbb{P}_{H_0}(R \le k) = \frac{k}{J+1} = \frac{\lfloor \alpha(J+1) \rfloor}{J+1} \le \alpha.
    \end{equation}
    
    Thus, $\mathbb{P}_{H_0}(p\text{-value} \le \alpha) \le \alpha$, proving that the permutation test is a valid level-$\alpha$ test in finite samples.
\end{proof}

\section*{Appendix D: Algorithmic Stability and Computational Cost}
\label{app:stability}

To address concerns regarding the reproducibility and stability of the WGAN training process, we conducted a sensitivity analysis based on the data generating process in Simulation 1. We specifically targeted the difficult scenario with 4\% heavy-tailed contamination to stress-test the estimator.

We varied the key hyperparameters of the Critic network: the network width ($d_{model}$), the learning rate ($\alpha$), and the gradient penalty coefficient ($\lambda_{GP}$). Table \ref{tab:sensitivity} reports the Average RMSE of the estimated weights $\hat{\lambda}$ and the Average Wall-Clock Execution Time over 20 Monte Carlo runs for each configuration.

\begin{table}[h]
\centering
\caption*{Table D1: Sensitivity of WGAN Estimator Performance and Cost to Hyperparameters}
\label{tab:sensitivity}
\begin{tabular}{lccccc}
\hline
\textbf{Configuration} & \textbf{Width} & \textbf{LR} & \textbf{GP Coeff} & \textbf{Avg RMSE} & \textbf{Avg Time (s)} \\ \hline
\textbf{Baseline (WGAN)} & \textbf{64} & $\mathbf{10^{-3}}$ & \textbf{10.0} & \textbf{0.0721} & \textbf{33.87} \\
\hline
\textit{Variation 1: Architecture} & 32 & $10^{-3}$ & 10.0 & 0.0712 & 30.19 \\
 & 128 & $10^{-3}$ & 10.0 & 0.0713 & 52.75 \\
\hline
\textit{Variation 2: Learning Rate} & 64 & $10^{-4}$ & 10.0 & 0.0706 & 33.74 \\
 & 64 & $10^{-2}$ & 10.0 & 0.0892 & 33.77 \\
\hline
\textit{Variation 3: Constraint} & 64 & $10^{-3}$ & 1.0 & 0.0715 & 33.79 \\
 & 64 & $10^{-3}$ & 0.1 & 0.0725 & 33.76 \\
\hline
\textbf{Benchmark ($L_2$-CDF)} & N/A & N/A & N/A & \textbf{0.2252} & \textbf{0.08} \\
\hline
\end{tabular}
\begin{flushleft}
\footnotesize
\emph{Note: Results averaged over 20 Monte Carlo runs with 4\% data contamination. The WGAN Baseline reduces the estimation error (RMSE) by a factor of 3 compared to the $L_2$ Benchmark.}
\end{flushleft}
\end{table}

The results in Table \ref{tab:sensitivity} support three key conclusions regarding the proposed method:

\begin{enumerate}
    \item \textbf{Parameter Stability:} The estimator exhibits a wide stable region. Varying the network width (32 to 128) or the gradient penalty coefficient (0.1 to 10.0) results in negligible fluctuations in RMSE (ranging between 0.071 and 0.072). Performance degradation is only observed when the learning rate is set aggressively high ($10^{-2}$), where RMSE increases to 0.0892, though it still significantly outperforms the benchmark.
    
    \item \textbf{Architecture Efficiency:} Our choice of the Baseline architecture (Width=64) represents an optimal trade-off. Increasing the width to 128 increases the computational cost by approximately 56\% (from 33.87s to 52.75s) without providing any improvement in estimation accuracy (RMSE 0.0713).
    
    \item \textbf{Cost-Benefit Analysis:} While the WGAN solver is computationally more intensive than the analytical $L_2$-CDF solver (approx. 34 seconds vs. 0.08 seconds), the performance gain is substantial. The WGAN approach reduces the RMSE from 0.2252 to 0.0721. In empirical research, where accuracy is paramount, trading less than one minute of computation time for a $3\times$ reduction in estimation bias is a highly favorable exchange.
\end{enumerate}

\section*{Appendix E}

\begin{figure}[H]
    \centering
    \includegraphics[width=1\linewidth]{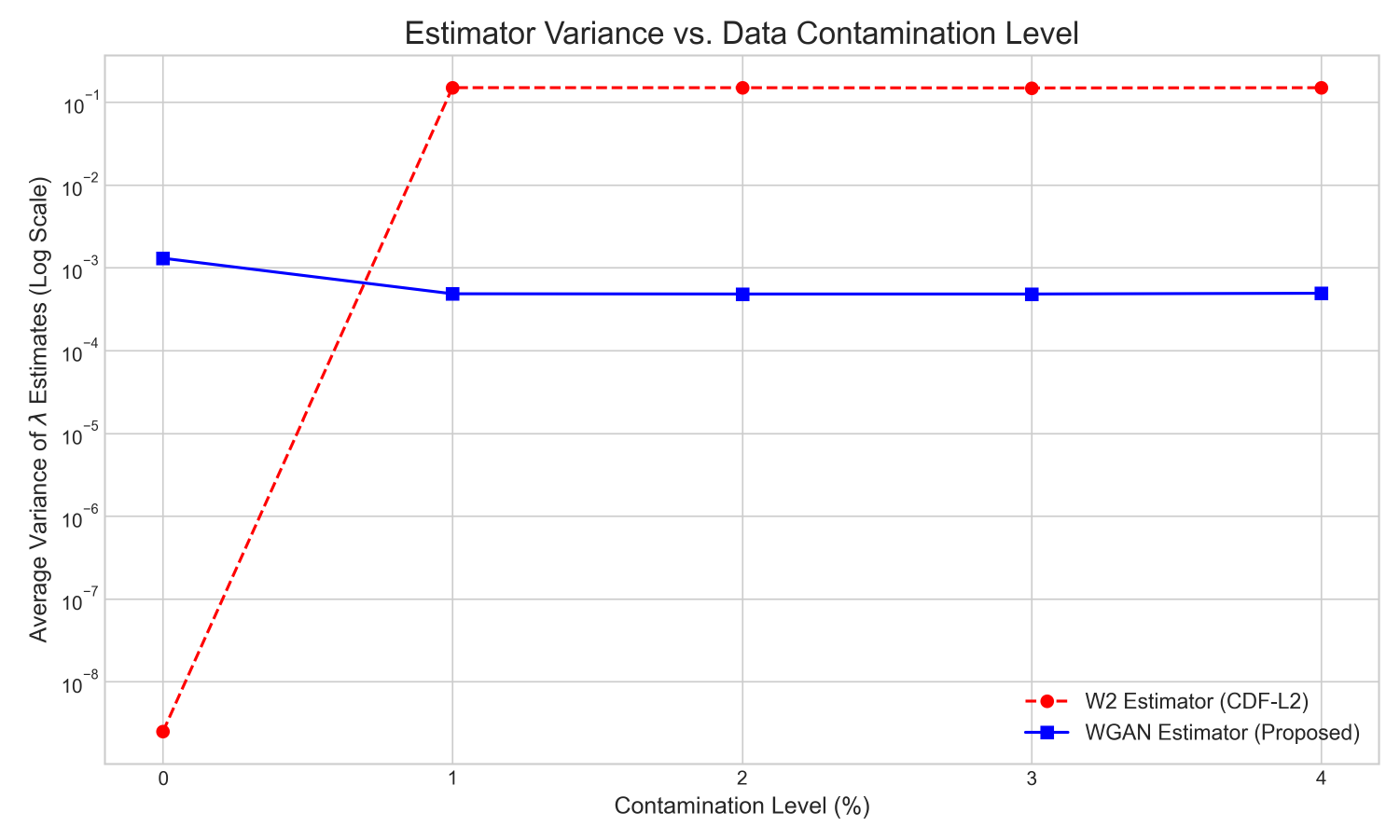}
    \caption*{Figure E1: The Average Variance of $\lambda$ over Different Contamination Levels}
    \label{fig:placeholder}
\end{figure}

\begin{landscape}

\begin{table}[htbp]
\centering
\caption*{Table E1: Bias / Var / W1 / W2 (per period, 1\%heavy-tail contaminated)}
\label{simulation2}
\begin{tabular}{llrrrrrrrrrr}
\toprule
\textbf{Period} & \textbf{Method} & \textbf{W2\_mean} & \textbf{W1\_mean} & \textbf{Bias\_$\lambda_1$} & \textbf{Var\_$\lambda_1$} & \textbf{Bias\_$\lambda_2$} & \textbf{Var\_$\lambda_2$} & \textbf{Bias\_$\lambda_3$} & \textbf{Var\_$\lambda_3$} & \textbf{Bias\_$\lambda_4$} & \textbf{Var\_$\lambda_4$} \\
\midrule
t1 & W2   & 0.061216 & 0.458148 & 0.068579 & 0.143471 &  0.077914 & 0.175395 & -0.073466 & 0.159012 & -0.073027 & 0.108486 \\
t1 & WGAN & 0.065782 & 0.493422 & 0.100645 & 0.000539 & -0.001229 & 0.000489 & -0.100268 & 0.000481 &  0.000853 & 0.000503 \\
t2 & W2   & 0.056452 & 0.394906 & 0.172024 & 0.185865 & -0.038184 & 0.131364 & -0.136211 & 0.143936 &  0.002371 & 0.168702 \\
t2 & WGAN & 0.061352 & 0.429676 & 0.099921 & 0.000478 & -0.000890 & 0.000464 & -0.102225 & 0.000429 &  0.003194 & 0.000531 \\
t3 & W2   & 0.053392 & 0.357901 & 0.103638 & 0.145807 & -0.008058 & 0.147179 & -0.092062 & 0.145742 & -0.003518 & 0.147599 \\
t3 & WGAN & 0.057128 & 0.383460 & 0.097526 & 0.000394 &  0.001314 & 0.000543 & -0.098108 & 0.000480 & -0.000732 & 0.000464 \\
\bottomrule
\end{tabular}
\end{table}

\begin{table}[htbp]
\centering
\caption*{Table E2: Bias / Variance / W2 (per period, 2\% heavy-tail contaminated)}
\label{simulation2}
\begin{tabular}{llrrrrrrrrr}
\toprule
\textbf{Period} & \textbf{Method} & \textbf{W2\_mean} & \textbf{Bias\_$\lambda_1$} & \textbf{Var\_$\lambda_1$} & \textbf{Bias\_$\lambda_2$} & \textbf{Var\_$\lambda_2$} & \textbf{Bias\_$\lambda_3$} & \textbf{Var\_$\lambda_3$} & \textbf{Bias\_$\lambda_4$} & \textbf{Var\_$\lambda_4$} \\
\midrule
t1 & W2   & 0.059265 &  0.068588 & 0.142982 &  0.074507 & 0.172140 & -0.074568 & 0.154759 & -0.068528 & 0.108766 \\
t1 & WGAN & 0.063801 &  0.100757 & 0.000541 & -0.001267 & 0.000491 & -0.100278 & 0.000483 &  0.000788 & 0.000503 \\
t2 & W2   & 0.054668 &  0.170705 & 0.186228 & -0.039308 & 0.131252 & -0.136150 & 0.142809 &  0.004753 & 0.166316 \\
t2 & WGAN & 0.059514 &  0.099932 & 0.000480 & -0.000838 & 0.000467 & -0.102283 & 0.000430 &  0.003190 & 0.000531 \\
t3 & W2   & 0.051781 &  0.105282 & 0.148106 & -0.004821 & 0.147294 & -0.093533 & 0.142721 & -0.006928 & 0.144321 \\
t3 & WGAN & 0.055495 &  0.098308 & 0.000400 &  0.001588 & 0.000541 & -0.098705 & 0.000481 & -0.001192 & 0.000466 \\
\bottomrule
\end{tabular}
\end{table}
\end{landscape}

\begin{landscape}
\begin{table}[htbp]
\centering
\caption*{Table E3: Bias / Variance / W2 under 3\% and 4\% contamination levels}
\label{3and4}

\begin{subtable}{0.95\linewidth}
\centering
\caption{3\% contamination}
\label{tab:3percent}
\resizebox{\linewidth}{!}{%
\begin{tabular}{llrrrrrrrrr}
\toprule
\textbf{Period} & \textbf{Method} & \textbf{W2\_mean} & \textbf{Bias\_$\lambda_1$} & \textbf{Var\_$\lambda_1$} & \textbf{Bias\_$\lambda_2$} & \textbf{Var\_$\lambda_2$} & \textbf{Bias\_$\lambda_3$} & \textbf{Var\_$\lambda_3$} & \textbf{Bias\_$\lambda_4$} & \textbf{Var\_$\lambda_4$} \\
\midrule
t1 & W2   & 0.057430 &  0.073791 & 0.146167 &  0.070292 & 0.170923 & -0.074066 & 0.153770 & -0.070017 & 0.107139 \\
t1 & WGAN & 0.061938 &  0.101016 & 0.000549 & -0.000845 & 0.000486 & -0.100637 & 0.000485 &  0.000466 & 0.000508 \\
t2 & W2   & 0.053103 &  0.172074 & 0.189955 & -0.040133 & 0.129689 & -0.134728 & 0.142373 &  0.002786 & 0.164408 \\
t2 & WGAN & 0.057888 &  0.099966 & 0.000480 & -0.000821 & 0.000468 & -0.102311 & 0.000432 &  0.003167 & 0.000532 \\
t3 & W2   & 0.050367 &  0.104191 & 0.148390 &  0.000439 & 0.150135 & -0.099174 & 0.139836 & -0.005456 & 0.142830 \\
t3 & WGAN & 0.054010 &  0.098292 & 0.000401 &  0.001601 & 0.000541 & -0.098698 & 0.000482 & -0.001196 & 0.000467 \\
\bottomrule
\end{tabular}%
}
\end{subtable}

\vspace{1em} 

\begin{subtable}{0.95\linewidth}
\centering
\caption{4\% contamination}
\label{tab:4percent}
\resizebox{\linewidth}{!}{%
\begin{tabular}{llrrrrrrrrr}
\toprule
\textbf{Period} & \textbf{Method} & \textbf{W2\_mean} & \textbf{Bias\_$\lambda_1$} & \textbf{Var\_$\lambda_1$} & \textbf{Bias\_$\lambda_2$} & \textbf{Var\_$\lambda_2$} & \textbf{Bias\_$\lambda_3$} & \textbf{Var\_$\lambda_3$} & \textbf{Bias\_$\lambda_4$} & \textbf{Var\_$\lambda_4$} \\
\midrule
t1 & W2   & 0.055836 &  0.071413 & 0.145747 &  0.072847 & 0.171297 & -0.077322 & 0.149187 & -0.066938 & 0.107604 \\
t1 & WGAN & 0.060274 &  0.101065 & 0.000550 & -0.000878 & 0.000487 & -0.100699 & 0.000489 &  0.000512 & 0.000511 \\
t2 & W2   & 0.051678 &  0.169040 & 0.189742 & -0.038132 & 0.129740 & -0.133904 & 0.141666 &  0.002996 & 0.162231 \\
t2 & WGAN & 0.056377 &  0.099986 & 0.000480 & -0.001200 & 0.000471 & -0.102193 & 0.000441 &  0.003407 & 0.000541 \\
t3 & W2   & 0.049146 &  0.104179 & 0.149000 &  0.000874 & 0.149375 & -0.103374 & 0.132732 & -0.001679 & 0.140943 \\
t3 & WGAN & 0.052708 &  0.098331 & 0.000401 &  0.001629 & 0.000542 & -0.098819 & 0.000482 & -0.001141 & 0.000469 \\
\bottomrule
\end{tabular}%
}
\end{subtable}

\end{table}
\end{landscape}

\end{document}